\documentclass[reqno,11pt]{amsart}
\usepackage[dvipsnames]{xcolor}
\usepackage{geometry,enumitem,comment,bbold}
\usepackage[normalem]{ulem}

\usepackage[utf8x]{inputenc}
\usepackage{amsmath,float}
\usepackage{graphicx}
\usepackage{rotating}
\usepackage{wasysym}
\DeclareGraphicsExtensions{.pdf,.png,.jpg}
\usepackage{setspace}

\usepackage{geometry}
\usepackage{booktabs}
\usepackage{amsmath}
\usepackage{amsthm}
\usepackage{amsfonts}
\usepackage{amssymb}
\usepackage{dsfont}

\usepackage{nicefrac}
\usepackage{booktabs}
\usepackage{mathrsfs}
\usepackage{bm} 
\usepackage{mathtools}   
\usepackage{subfigure}

\newtheorem{theorem}{Theorem}[section]
\newtheorem{corollary}[theorem]{Corollary}      
\newtheorem{lemma}[theorem]{Lemma}              
  
\newtheorem{proposition}[theorem]{Proposition}  
\theoremstyle{definition}
\newtheorem{example}[theorem]{Example} 
\newtheorem{definition}[theorem]{Definition} 
\newtheorem{remark}[theorem]{Remark}
\newtheorem{assumption}[theorem]{Assumption}

\numberwithin{equation}{section}

\DeclareMathOperator{\es}{\mathrm{ES}}
\DeclareMathOperator{\epi}{\mathrm{epi}}
\DeclareMathOperator{\epiG}{\epi^{G}}
\DeclareMathOperator{\interior}{\mathrm{int}}
\DeclareMathOperator{\essinf}{\mathrm{ess\,inf}}
\DeclareMathOperator{\e}{\mathrm{e}}

\newcommand{\R}{\mathbb{R}} 
\newcommand{\Q}{\mathbb{Q}}

\newcommand{\E}{\mathbb{E}}

\renewcommand{\P}{\mathbb{P}}

\newcommand{\N}{\mathbb{N}}
\newcommand{\dom}{\textnormal{dom}}

\newcommand{\ba}{\mathbf{ba}}

\renewcommand{\Q}{{\mathbb Q}}

\newcommand*\diff{\mathop{}\!\mathrm{d}}

\newcommand{\riskM}{\rho}
\newcommand{\rrm}{\widetilde{\riskM}}
\DeclareMathOperator{\core}{core}

\DeclareMathOperator*{\esssup}{ess\,sup}
\DeclareMathOperator{\unit}{\mathbf{e}}
\newcommand{\ball}{B}

\DeclareMathOperator{\disjointComplement}{d}
\newcommand{\dualElementGG}{l}
\newcommand{\dualElementClassic}{\psi}

\newcommand{\aggregation}{\tilde{\Lambda}}
\newcommand{\normAM}[1]{\lVert #1 \rVert_{\unit}}

\usepackage{natbib}
\usepackage[colorlinks,urlcolor=red,citecolor=blue,linkcolor=red]{hyperref}

\usepackage{orcidlink}

\begin{document}

\title[Geometrically convex return risk measures on AM-algebras]{Geometrically convex return risk measures on AM-algebras}

\author{Christian Laudag\'e \orcidlink{0000-0002-5910-9487}}
\address{Department of Mathematics, RPTU University Kaiserslautern-Landau, Germany.}
\email{christian.laudage@rptu.de}

\date{\today}

\begin{abstract}
\setstretch{1.1}

Monetary risk measures quantify the risk of uncertain monetary payoffs (or losses), whereas in time series analysis risk is typically assessed using logarithmic returns. 
Return risk measures (RRMs) provide an axiomatic foundation for the latter approach, which relies crucially on the positive cone of the space of essentially bounded random variables. 
We extend RRMs to general ordered vector spaces and characterize positive homogeneity via the geometric epigraph. To investigate geometric convexity and establish connections with monetary risk measures, we specialize the domain as a subset of an AM-algebra, encompassing Euclidean spaces and spaces of multidimensional essentially bounded random variables. The latter is novel in the context of RRMs and leads to the new classes of systemic and vector-valued RRMs. We establish results on finiteness, continuity, separability, as well as dual and aggregation-based representations.

\smallskip

\noindent{\em Keywords:} AM-algebras, monetary risk measures, return risk measures, separability, systemic risk, vector-valued risk measures

\end{abstract}

\maketitle
\thispagestyle{empty} 

\setstretch{1.2}    
\section{Introduction}\label{section.intro}

Monetary risk measures have become a standard tool in finance and actuarial science. While their original purpose --- evaluating the risk of a monetary payoff or loss --- is well understood, an axiomatic framework for their application to logarithmic payoffs, losses or returns was not available until the introduction of return risk measures (RRMs) in~\cite{bellini2018RRM}. Geometrically convex (GG-convex) RRMs have been recently studied in~\cite{ayguen}. Whereas convex monetary risk measures reward diversification among buy-and-hold trading strategies, GG-convex RRMs reward diversification among continuously rebalanced trading strategies. In this manuscript, we contribute by studying multivariate extensions --- both in the domain and the codomain --- of GG-convex RRMs.

Moreover, our approach unifies earlier works on RRMs, namely approaches focusing on $L^{\infty}$ or~$(0,\infty)$ as domains, while simultaneously extending it to the multivariate framework. This is achieved by using a subset of an AM-algebra as domain. Although the application of AM-algebras in finance is rather limited, they provide a natural setting for our purposes and lead to new insights into risk measurement based on logarithmic transformations of the underlying quantities. 

We now recall the state-of-the-art on RRMs. In a dynamic time framework they have been studied in~\cite{zullino_BSDEs}. RRMs using multiple eligible assets to hedge uncertain losses are defined and analyzed in~\cite{laudage2025MultiAssetRRM}. Another framework for RRMs that goes beyond $L^{\infty}$ is the one in~\cite{zullino2025IME}. However, this work focuses on locally convex topological vector spaces between $L^{\infty}$ and $L^1$, which do not exploit the specific structure of AM-algebras. Moreover, their main intention is to obtain representation results for stochastically consistent (but not necessarily cash-additive) functionals. In contrast, we focus on GG-convexity and its connection to monetary risk measures. In addition, we clarify several important technicalities that are essential for working with RRMs at a deeper mathematical level. For completeness, note that GG-convexity has also been studied beyond finance, see e.g.,~\cite{niculescu2000,niculescuPersson2004,oezdemir2014}.

Important examples of multivariate risk measures are systemic risk measures, which are axiomatically characterized on finite state spaces in~\cite{chen2013SystemicRiskMeasures} and later extended to locally convex-solid Riesz spaces in~\cite{kromer2016}. Conditional versions of systemic risk measures can be found in~\cite{hoffmann2016SystemicRiskMeasures}. Inspired by these contributions, we introduce systemic RRMs. Finally, together with the creation of systemic optimal allocations in~\cite{Biagini2020SystemicRisk} via the dual representations in~\cite{Biagini2019SystemicRisk}, the results in~\cite{AraratFeinstein2024Separability} motivate the study of vector-valued risk measures, a concept between scalar and set-valued risk measures. Following their ideas, we discuss separability results for vector-valued RRMs. 

In this manuscript, AM-algebras are used as a small domain on which an RRM can admit properties like GG-convexity or representations via monetary risk measures, while at the same time being sufficiently rich to support a multivariate RRM framework. This approach has several advantages. First, AM-algebras unify different RRM frameworks in the literature, thereby clarifying structural features of RRMs and eliminating the need to develop analogous results for special cases separately. Second, AM-algebras are general enough to encompass more demanding domains, such as families of vector-valued random variables. Third, developing a theory of RRMs on AM-algebras helps to clarify both the possibilities and limitations of applying RRMs. Finally, there is also an intrinsic mathematical motivation for developing a unified framework for RRMs.

Our study addresses the following three main topics:
\begin{itemize}
    \item \textbf{GG-convex RRMs on AM-algebras:} We introduce a new unifying (multivariate) framework for RRMs and perform a corresponding analysis of it. This framework consists of the definition of GG-convex RRMs on AM-algebras. Then, we characterize the defining properties of positive homogeneity and GG-convexity via the geometric epigraph. Using the connection to convex monetary risk measures, we prove that GG-convex RRMs are only locally Lipschitz-continuous in general. Then, we establish two dual representations, with a dual maximizer guaranteed only for the first one. 
    \item \textbf{GG-convex systemic RRMs:} The most important AM-algebra in finance is the space of essentially bounded random vectors, i.e.,~we aim to evaluate the risk of a vector of losses, which is closely related to measure systemic risk. Therefore, we introduce systemic RRMs, which map vector-valued losses to positive real numbers. We show that this is equivalent to the existence of an RRM and a return aggregation function such that their composition yields the systemic RRM. We illustrate that this decomposition crucially relies on the assumption of preference consistency. Moreover, a GG-convex and positively homogeneous systemic RRM satisfies risk GG-convexity, a property closely related to GG-convexity. Without positive homogeneity, a counterexample shows that GG-convexity of a systemic RRM is not sufficient for risk GG-convexity. We conclude with a dual representation for GG-convex systemic risk measures, which clarifies that the impact of the RRM and the return aggregation function can be fully separated from each other.   
    \item \textbf{Separability of GG-convex RRMs:} Now, we study RRMs with $(0,\infty)^n$ as codomain. In this context, separability means that the vector-valued evaluation can be reduced to single-valued evaluations; that is, the $i$-th component of the RRM depends only on the $i$-th loss. We derive sufficient conditions for such a reduction. Our first result is obtained under pointwise positive homogeneity, which is stronger than positive homogeneity. Under this assumption, RRMs defined on the interior of the cone of positive, essentially bounded random vectors are separable. Finally, we extend these results to larger domains and present examples based on well-known monetary risk measures, namely Expected Shortfall and entropic risk measures.
\end{itemize}

The structure of the manuscript is as follows. In Section~\ref{sec:unifyingFramework}, we introduce the unifying framework based on AM-algebras and define GG-convexity. Continuity properties and dual representations are studied in Section~\ref{sec:continuity}. Section~\ref{sec:multivariateRRM} is devoted to systemic risk. Here, we specialize to vectors of positive, essentially bounded random variables and $(0,\infty)$ as codomain. 
In Section~\ref{sec:vectorValuedRRM}, we study the separability of GG-convex RRMs in the vector-valued setting, with codomain $(0,\infty)^n$.
In Section~\ref{sec:conclusionOutlook}, we conclude and outline two directions for future research.

\section{Return risk measures and AM-algebras}\label{sec:unifyingFramework}

Domains for RRMs appearing in the current literature are subsets of positive cones in $L^p$-spaces. Nonetheless, risk measure theory has an inherent tendency towards unifying frameworks,  allowing for more complex settings, such as stochastic processes to account for risk evaluation over entire trajectories. The unification we present below covers the cases of $(0,\infty)^n$ and $L^{\infty}(\Omega;\mathbb{R}^n)$. 

\subsection{Standard notation and assumptions}

Before turning to the main results, we recall some standard notation and assumptions. This subsection may be skipped and consulted later if needed. Let $V$ be a linear space and $C\subseteq V$. Then, $C$ is convex, if for all $\lambda\in(0,1)$ it holds that $\lambda C+(1-\lambda) C\subseteq C$. The set $C\subseteq V$ is a cone, if for all $\lambda\in(0,\infty)$ it holds that $\lambda C\subseteq C$.

The usual Euclidean norm for a point $x\in\mathbb{R}^n$ is denoted by $|x|$. For $x,y\in\mathbb{R}^n$, we denote by $x\cdot y$ (or $xy$ for short) the pointwise multiplication. Furthermore, we always equip $\mathbb{R}^n$ with the pointwise ordering, denoted by $\leq$. For brevity, we also use the notation $[n]:=\{1,\dots,n\}$ for all $n\in\N$. For $x,y\in\R$, the notation $x<y$ means that for all $i\in[n]$ it holds that $x_i<y_i$.

Let $(X,\tau)$ be a topological space. We write $x_\alpha\xrightarrow{\tau} x$, when the net $(x_\alpha)_\alpha$ in $X$ converges to the point $x\in X$. In a normed space $(V,\lVert.\rVert)$, the ball with radius $r>0$ around $x\in V$ is defined by $\ball_{\lVert.\rVert}(x,r):=\{y\in V\,\mid\, \lVert x-y \rVert < r\}$. The interior of a subset $A\subseteq V$ is denoted by $\interior_{\lVert.\rVert}(A)$. The (topological) dual of $V$ and its bidual are denoted by $V^{\prime}$ and $V^{\prime\prime}:=(V^{\prime})^{\prime}$, respectively. The (upper) support function of $E\subseteq V$ at $\dualElementClassic\in E$ is given by $\sigma_E(\dualElementClassic):=\sup\limits_{x\in E} \dualElementClassic(x)$. For a dual pair $\langle X,Y\rangle$, we denote the weak topology by $\sigma(X,Y)$.  

Assume a nonempty set $V$ and a map $f:V\rightarrow [-\infty,\infty]^m$. The epigraph of $f$ is defined by $\epi(f):=\{(x,\alpha)\in V\times \R^m\mid f(x)\leq \alpha\}$. The domain of $f$ is $\dom(f):=\{x\in V\mid{\forall i\in[m]:f_i(x)<\infty}\}$. Note, $(x,f(x))\in \epi(f)$ implies that $x\in\dom(f)$. For $m=1$, the map~$f$ is proper, if $f>-\infty$ and $\dom(f)\neq\emptyset$. If $X\subseteq V$, then we say that $f$ is convex if $\epi(f)$ is convex. If $f>-\infty$, then convexity of $f$ is equivalent to $f\big(\lambda x + (1-\lambda)y\big)\leq \lambda f(x) + (1-\lambda) f(y)$ for all~$\lambda \in(0,1)$ and~$x,y\in X$. 

We always impose a suitable\footnote{For some of our examples, we require the existence of a continuous random variable, which can, e.g.,~be ensured by assuming that the probability space is atomless, see~\cite[Proposition A.31]{FoeSch}.} probability space, denoted by $(\Omega,\mathcal{F},\mathbb{P})$. The linear space of corresponding $n$-dimensional random variables is denoted by $L^0(\Omega;\mathbb{R}^n)$.
For $p\in[1,\infty]$, $L^p(\Omega;\mathbb{R}^n)$, or simply $L^p$ when $n=1$, is the linear space of all $\mathbf{X}\in L^{0}(\Omega;\mathbb{R}^n)$ with~$\lVert\mathbf{X}\rVert_{p}:=\big(\E\big[\sum_{i=1}^n|\mathbf{X}_i|^p\big]\big)^{1/p}<\infty$ if $p\in[1,\infty)$ and $\lVert\mathbf{X}\rVert_{\infty} :=\max\limits_{i\in[n]}\,\esssup |\mathbf{X}_i|<\infty$ if $p=\infty$. The corresponding set of random variables which are strictly positive a.s.~is $L_{++}^{p}(\Omega;\mathbb{R}^n):=\{\mathbf{X}\in L^{p}(\Omega;\mathbb{R}^n)\mid \mathbf{X}>0\text{ a.s.}\}$ or simply $L^p_{++}$ when $n=1$.\footnote{The reader should not confuse this set with the set of strictly positive elements with respect to the standard $L^p$-norm.}

For arbitrary $n\in\mathbb{N}$, $x\in\mathbb{R}^n$ and $g:\mathbb{R}\rightarrow\mathbb{R}$, we define $g(x):=(g(x_1),\dots,g(x_n))^{\intercal}$, i.e.,~if not otherwise stated, we always apply functions pointwise. For $x\in[0,\infty)^n$ and $\lambda\in[0,\infty)$ we write $x^{\lambda}:=((x_1)^{\lambda},\dots,(x_n)^{\lambda})^{\intercal}$. We always work under the standard conventions $ \infty\cdot 0 = 0$ and $(-\infty)\cdot 0 = 0$. 

\subsection{Ordered vector spaces and return risk measures}

The concepts presented next are standard and taken from~\cite{Rudin} and~\cite{AliprantisBurkinshaw}. First of all, in addition to a vector space structure, we need a partial order. A pair $(V,\leq)$, where $V$ is a real vector space and $\leq$ is a partial order, is an \textbf{ordered vector space}, if the partial order $\leq$ is compatible with the algebraic structure of $V$, i.e.,~for all $x,y,z\in V$ with $x\leq y$ and~$\alpha\in[0,\infty)$ it holds that 
\[
    \alpha x+z\leq \alpha y+z.
\]
For brevity, we often write $V$ instead of $(V,\leq)$. The positive cone of $V$ is defined by 
\[
    V_{+}:=\{x\in V\mid0\leq x\}.
\]

Now we define a return risk measure, for which it is enough to work on an ordered vector space. 
\begin{definition}\label{defi:RRM}
    Let $V$ be an  ordered vector space. A map $\rrm:E\rightarrow[0,\infty]^{m}$ with $E\subseteq V_+$ being a cone, is called a \textbf{return risk measure (RRM)}, if it satisfies the following two properties:
    \begin{enumerate}
        \item[(1)] For all $x,y\in E$ with $x\leq y$ it holds that $\rrm(x)\leq \rrm(y)$ \textit{(monotonicity)};
        \item[(2)] For all $x\in E$ and $\lambda\in(0,\infty)$ it holds that $\rrm(\lambda x) = \lambda \rrm(x)$ \textit{(positive homogeneity)}.
    \end{enumerate}
    The corresponding acceptance set is $\mathcal{B}_{\rrm} :=\{x\in E\mid \rrm(x)\leq 1\}$.
\end{definition}

We interpret elements in $E$ as pure losses. Therefore,  monotonicity means that an almost surely larger loss than another one results in a larger RRM value, i.e.,~a larger capital reserve. The assumption of positive homogeneity is motivated by its connection to the defining property of cash-additivity for monetary risk measures, see Proposition~\ref{prop:relationToClassicalRM} below. Positive homogeneity means that scaling a loss before applying the RRM yields the same result as scaling the RRM of that loss.

Properties like positive homogeneity and convexity for monetary risk measures with codomain $[-\infty,\infty]$ are defined via the epigraph. The epigraph excludes $-\infty$ by definition. Analogously, this would mean to exclude zero for RRMs, which is not obtained by applying the epigraph itself. Therefore, it is important to also study the notion of a geometric epigraph and compare its properties with those of the classical epigraph and our definition of positive homogeneity for an RRM. 
\begin{definition}\label{defi:geometricEpigraph}
    Let $V$ be an ordered vector space and $E\subseteq V_+$. The \textbf{geometric epigraph} of a map $f:E\rightarrow [0,\infty]^{m}$ with $m\in\N$ is 
    \begin{align}\label{eq:connection_epi_epiG}
        \epiG(f):=\epi(f)\cap \big(E\times (0,\infty)^m\big) = \{(x,\alpha)\in E\times (0,\infty)^{m}\mid f(x)\leq\alpha\}.
    \end{align}
\end{definition}

\begin{remark}\label{rem:epiG}
    For $m=1$, Equation~\eqref{eq:connection_epi_epiG} shows for $f:E\rightarrow [0,\infty]^{m}$ that
    \begin{align}\label{eq:epiG_rem}
        \epi(f)\setminus\epiG(f) = \{x\in E \mid f(x) = 0\}\times\{0\}.
    \end{align}
    Since $(x,f(x))\in\epi(f)$ iff $x\in\dom(f)$, the previous observation yields 
    \[
        (x,f(x))\in\epiG(f)\implies x\in\dom(f).
    \]
\end{remark}

The inclusion $\epiG(f)\subseteq\epi(f)$ can be strict as the following example shows.
\begin{example}\label{exam:epiG}
    Let $f:(0,\infty)^{2}\rightarrow [0,\infty),x\mapsto \max\{x_2-x_1,0\}$. Then, by~\eqref{eq:epiG_rem} we obtain that  
    \[
        \epi(f)\setminus\epiG(f) = \big\{x\in (0,\infty)^2\,\big|\, x_2\leq x_1\big\}\times\{0\}.
    \]
\end{example}

Next, we analyze the relation between positive homogeneity of a map $f:E\rightarrow [0,\infty]^m$ and conicity of its (geometric) epigraph. Since, additional assumptions are needed to ensure strict positivity of an RRM, see e.g.,~the discussion in~\cite[Section 4.2]{laudage2025MultiAssetRRM}, we use in the upcoming result $[0,\infty)^m$  --- instead of $(0,\infty)^m$ --- as codomain. 

\begin{proposition}\label{prop:positiveHomogeneity}
    Let $V$ be an  ordered vector space. For a map $f:E\rightarrow[0,\infty]^m$ with $E\subseteq V_+$ being a cone, we consider the following statements:
    \begin{enumerate}
        \item[(a)] $f$ is positively homogeneous;
        \item[(b)] $\epi(f)$ is a cone;
        \item[(c)] $\epiG(f)$ is a cone.
    \end{enumerate}
    Then, $(a)\implies (b)\Longleftrightarrow (c)$. If $m = 1$ or $f:E\rightarrow[0,\infty)^m$, then it holds that $(b)\implies (a)$. 
\end{proposition}

\begin{proof}
     We first prove that (a) implies (b). Let $(x,\alpha)\in\epi(f)$. Hence, $f(x)\leq \alpha$. Then by (a), for all $\lambda > 0$ we obtain 
    \[
        f(\lambda x) = \lambda f(x)\leq \lambda \alpha\in[0,\infty)^{m}.
    \]
    Hence, $\lambda (x,\alpha) = (\lambda x,\lambda \alpha)\in\epi(f)$, which shows that $\epi(f)$ is a cone. For proving that (b) implies (c), let $\lambda>0$ and $(x,\alpha)\in\epiG(f)$. Note that $\alpha\in(0,\infty)^{m}$ implies $\lambda\alpha\in(0,\infty)^m$. Hence, $(\lambda x,\lambda\alpha)\in\epi(f)\cap (E\times(0,\infty)^m) = \epiG(f)$, where we make use of~\eqref{eq:connection_epi_epiG}.

    Next, we show that (c) implies (b). To do so, let $\lambda>0$ and $(x,\alpha)\in \epi(f)$. If $\alpha\in(0,\infty)^m$, then $(x,\alpha)\in\epiG(f)$ and by assumption $\lambda(x,\alpha)\in\epiG(f)\subseteq\epi(f)$. Otherwise, assume $\alpha\notin (0,\infty)^m$, which implies that $I:=\{i\in[m]\mid\alpha_i = 0\}\neq \emptyset$. Then, for all $i\in I$ we have $f_i(x) = 0$.     Recall, that for an arbitrary $\beta\in(0,\infty)^m$ with $(x,\beta)\in\epiG(f)$, by assumption it holds for all $\lambda>0$ that  $(\lambda x,\lambda\beta)\in\epiG(f)$. Then, we obtain for all $i\in I$ that
    \[
        f_i(\lambda x)\leq \inf_{(x,\beta)\in\epiG(f)}\lambda \beta_i = 0. 
    \]
    This implies that $f_i(\lambda x)=0$ for all $i\in I$. Hence, $\lambda (x,\alpha)\in\epi(f)$, showing that $\epi(f)$ is a cone.
    
    Finally, for $m = 1$ or $f:E\rightarrow[0,\infty)^m$ we show that (b) implies (a). We first consider the case of $f:E\rightarrow[0,\infty)^m$, for which we assume $\lambda>0$ and $x\in E$. Then, $(x,f(x))\in\epi(f)$ and by assumption $(\lambda x,\lambda f(x))\in \epi(f)$, which is equivalent to 
    \[
        f(\lambda x)\leq \lambda f(x).
    \]
    To obtain the inverse inequality, note that $(\lambda x, f(\lambda x))\in \epi(f)$ and hence, $(x,\lambda^{-1}f(\lambda x))\in \epi(f)$, which means that
    \[
        \lambda f(x)\leq f(\lambda x).
    \]

    For the remaining case of $m=1$, let $x\in E$ and $\lambda>0$. If $f(x)=\infty$, then $f(\lambda x)=\infty$. Indeed, if $f(\lambda x)\in[0,\infty)$, then $(\lambda x,f(\lambda x))\in\epi(f)$ and by assumption implying that $(x,\lambda^{-1}f(\lambda x))\in \epi(f)$. However, the latter means that $\infty = f(x)\leq \lambda^{-1}f(\lambda x)<\infty$, a contradiction. In the same manner, we obtain that $f(\lambda x) =\infty$ implies $f(x) = \infty$. The remaining case of $f(x) \in[0,\infty)$, which implies $f(\lambda x)\in[0,\infty)$, follows analogously to the previous considerations for $f:E\rightarrow[0,\infty)^m$.    
\end{proof}

\begin{remark}
    Note, $\epiG(f)$, and hence also $\epi(f)$, in Example~\ref{exam:epiG} is a cone.
\end{remark}

The counterexample below shows that the implication ``$(b)\implies (a)$'' may fail without the additional assumptions in Proposition~\ref{prop:positiveHomogeneity}, i.e.,~for $f(E)\not\subseteq [0,\infty)^m$ with $m>1$, positive homogeneity for RRMs as we defined it in Definition~\ref{defi:RRM} may differ from the epigraph being a cone. 
\begin{example}
    Let $f:[0,\infty)\rightarrow [0,\infty]^2$ with
    \[
        f(x) = \begin{cases}
            (0,0)^{\intercal}, & x = 0,\\
            (\infty,0)^{\intercal}, & x\in(0,1],\\
            (\infty,1)^{\intercal}, & x\in(1,\infty).            
        \end{cases}
    \]
    Note, $\epi(f) = \{0\}\times\{(0,0)^{\intercal}\}$ and hence, $\epi(f)$ is a cone. However, for each $\lambda >1$ it holds that 
    \[
        \lambda f(1) = (\infty,0)^{\intercal} \neq (\infty,1)^{\intercal} = f(\lambda). 
    \]
\end{example}

\subsection{AM-algebras and geometric convexity}

Next, we aim to define geometric convexity, which relies on powers and products of elements in the positive cone of an ordered vector space. For this, additional structure on the ordered vector space is required. As a first step, recall that an ordered vector space $V$ is called a \textbf{Riesz space} if, for every $x,y\in V$, the infimum and the supremum of the set~$\{x,y\}$ exist in $V$. In this case, for all $x,y\in V$ we write
\[
    x\vee y:=\sup\{x,y\},\quad x\wedge y:=\inf\{x,y\},
\]
and
\[
    x^{+}:=x\vee 0,\quad x^{-}:=(-x)\vee 0,\quad |x|:=x\vee (-x).
\]
A helpful representation holds for all $x,y\in V$~\cite[Theorem 1.2 (2)]{AliprantisBurkinshaw}:
\begin{align}\label{eq:sumElementsRieszSpace}
    x+y = x\wedge y + x\vee y.
\end{align}
This implies
\[
    x = x\wedge y + x\vee y - y = x\wedge y + (x-y)^{+}.
\]

A vector $\unit>0$ in a Riesz space $V$ is called an \textbf{order unit}, or simply a \textbf{unit}, if for every $x\in V$ there exists~$\lambda\in (0,\infty)$ such that $|x|\leq \lambda \unit$.\footnote{Equivalently, $\unit$ is an element of the algebraic interior $\core(V_{+})$ of $V_{+}$.} 

A norm $\lVert.\rVert$ on a Riesz space $V$ is called a \textbf{lattice norm}, if for every $x,y\in V$ with~$|x|\leq |y|$ it holds that $\lVert x\rVert\leq \lVert y\rVert$. A complete normed Riesz space is then called a \textbf{Banach lattice}, see~\cite[Section 9.1]{AliprantisBorder}. Then, a Banach lattice $V$ is an \textbf{AM-space}, if $\lVert x\vee y\rVert\leq \lVert x\rVert \vee \lVert y\rVert$ for all $x,y\in V_+$, see~\cite[Section 9.5]{AliprantisBorder}. Following~\cite[Section 9.5]{AliprantisBorder}, we say that $V$ is an \textbf{AM-space with unit $\unit$}, if $V$ is a Banach lattice with unit~$\unit$, equipped with the supremum norm $\normAM{.}$, defined by 
\[
    \normAM{x} := \inf\{\lambda\in(0,\infty)\mid\ |x|\leq \lambda \unit\},\quad x\in V.
\]

Then, we need a multiplication operator on our domain. For this, recall that a (real) \textbf{algebra} is a vector space $V$ over a real field on which an associative and distributive multiplication is defined, that is, for all $x,y,z\in V$ it holds that 
\[
    x(yz) = (xy)z,\quad (x+y)z = xz + yz,\quad x(y+z) = xy+xz,
\]
and multiplication is consistent to scalar multiplication, i.e.,~for all $x,y\in V$ and $\alpha\in\mathbb{R}$ it holds that
\[
    \alpha(xy) = x(\alpha y) = (\alpha x)y.
\]
A Riesz space $V$ with a multiplication is a \textbf{Riesz algebra}, if it becomes an algebra and in addition, for all $x,y\in V_{+}$ it holds that $xy\geq 0$. We summarize two properties of Riesz algebras. Further background on Riesz algebras can be found in~\cite{AliprantisBurkinshaw}.

\begin{lemma}
    Let $V$ be a Riesz algebra. The following statements hold:
    \begin{enumerate}
        \item[(i)] For all $x,y,z\in V_{+}$ with $x\leq y$ it holds that $xz\leq yz$ and $zx\leq zy$.
        \item[(ii)] For all $n\in\N$ and $x,y\in V_{+}$ with $x\leq y$ it holds that $x^n\leq y^n$.
    \end{enumerate}
\end{lemma}

\begin{proof}
    Let $x,y\in V_+$ with $x\leq y$ and set $w:= y-x\geq 0$. Then, for all $z\in V_{+}$ it holds that 
    \[
        yz-xz = wz \geq 0,
    \]
    which implies $xz\leq yz$. The inequality $zx\leq zy$ follows analogously, which proves (i). Then, (ii) follows by a proof of induction. Indeed, by assumption $x^1 = x\leq y = y^1$. Assuming $x^n\leq y^n$ for $n\in\N$, it holds that $x^{n+1} = x^nx \leq x^n y\leq y^n y = y^{n+1}$.
\end{proof}

An algebra $A$ equipped with a norm $\lVert.\rVert$ is called a \textbf{Banach algebra}, see~\cite[Definition~18.1]{Rudin}, if the norm is complete and satisfies 
\[
    \lVert xy\rVert\leq\lVert x\rVert \lVert y\rVert,\quad x,y\in A.
\]
Then, a Riesz algebra $V$ is a \textbf{Banach lattice algebra} if it is both a Banach lattice and a Banach algebra, see~\cite[Chapter IV, Exercise 4]{Schaefer1974BanachLatticeAlgebras} or~\cite{Martignon1980BanachLatticeAlgebras}. Our formulation follows~\cite{Wickstead2017}. 
Finally, $V$ is an \textbf{AM-algebra}, if it is a Banach lattice algebra which is an AM-space with order unit $\unit$ and $\unit$ is also an algebraic identity. 

AM-algebras provide a unified framework for the two most common settings in which geometrically convex RRMs are studied, namely  $\R^n$ and $L^{\infty}(\Omega,\R^n)$ (equipped with the supremum norm). To the best of our knowledge, the existing literature on geometrically convex RRMs considers only the case of $n=1$ and moreover treats these two settings separately. 
\begin{example}
    The following are AM-algebras:
    \begin{enumerate}
        \item $\mathbb{R}^n$ equipped with the pointwise multiplication and the maximum norm; 
        \item The space of essentially bounded random variables $L^{\infty}(\Omega,\mathbb{R}^n)$ with pointwise multiplication and sup-norm given by $\lVert.\rVert_{\infty}$. 
        \item The linear space of essentially bounded stochastic processes on a filtered probability space $\big(\Omega,\mathcal{F},(\mathcal{F}_t)_{t\in[0,T]},\mathbb{P}\big)$ with $T\in(0,\infty)$ and equipped with pointwise multiplication and essential supremum norm, i.e.,
        \[
            \big\lVert (X_t)_{t\in[0,T]}\big\rVert = \esssup_{(\omega,t)\in \Omega\times [0,T]} |X_t(\omega)|.
        \]
    \end{enumerate}    
    Note that $L^{p}(\Omega,\mathbb{R}^n)$ with $p\in[1,\infty)$ is not an algebra based on pointwise multiplication, because $X,Y\in L^1_+$ does not imply $XY\in L^1_+$ in general. Moreover, the standard $L_p$-norm does not satisfy the property of an AM-space, i.e.,~it does not hold for all $X,Y\in L^p$ that $\lVert X\vee Y\rVert_{p}\leq \lVert X\rVert_{p} \vee \lVert Y\rVert_{p}$.
\end{example}

On our way to define geometric convexity, we need to know if powers of positive elements in an AM-algebra are well-defined. This is indeed the case, as the next result shows. For its proof, we use the standard notation $C(X)$ to denote the collection of continuous functions on a subset $X$ of a topological space. If $X$ is a compact Hausdorff space, then $C(X)$ is an AM-space and even an AM-algebra when equipped with its usual pointwise multiplication.
\begin{proposition}
    Let $V$ be an AM-algebra with unit $\unit$. Then, for all $x\in V_+$ and $t\in[0,\infty)$, $x^t$ exists and lies in $V_+$. Moreover, the following intuitive properties hold:
    \begin{enumerate}
        \item[(i)] $x^{ts} = (x^t)^s$, \quad $x\in V_+$, $t,s\in[0,\infty)$;
        \item[(ii)] $x^{t} z^{t} = (xz)^t$, \quad  $x,z\in V_+$, $t\in[0,\infty)$;
        \item[(iii)] $x^{t} x^{s} = x^{t+s}$, \quad $x\in V_+$, $t,z\in[0,\infty)$.
    \end{enumerate}
\end{proposition}

\begin{proof}
    We apply a version of the Krein-Kakutani Theorem for AM-algebras, which goes back to~\cite{Martignon1980BanachLatticeAlgebras}.\footnote{The original Krein-Kakutani Theorem is stated for AM-spaces and not for AM-algebras, see~\cite{Kakutani1941}.} It states that the AM-algebra $V$ with unit $\unit$ is lattice and algebraic isometric to~$C(K)$ for some compact Hausdorff space $K$ with unit $\mathbf{1}_K$. By $\mathbf{1}_K$ being the unit in~$C(K)$, we can apply~\cite[Proposition 1.4]{Martignon1980BanachLatticeAlgebras} to obtain that the algebraic operation on~$C(K)$ has to be the pointwise multiplication. We denote the corresponding surjective lattice and algebraic isometry by $\Phi:V\rightarrow C(K)$ and for short, we also write $\hat{x} = \Phi(x)$. Then, for $x\in V_+$ and $y\in[0,\infty)$, the expression $(\hat{x})^t$ is well-defined in the pointwise sense and it is contained in $C(K)$. Indeed, by using that $\Phi$ is a lattice homomorphism we can apply~\cite[Theorem 9.15 (3)]{AliprantisBorder} to obtain $\hat{x} = \Phi(x)= \Phi(x^{+})=\Phi(x)\vee 0 = \hat{x}^+\geq 0$. Moreover, $(\hat{x})^t$ as composition of continuous functions is again continuous. Hence, $(\hat{x})^t\in C(K)$. Then, set $x^{t} = \Phi^{-1}((\hat{x})^t)$ and analogously to before we obtain that~$x^{t}\in V_+$.

    To prove (i), let $x\in V_+$ and $t,s\in [0,\infty)$. By 
    \[
        \widehat{x^t} = \Phi(x^t) = \Phi(\Phi^{-1}(\hat{x}^t)) =\hat{x}^t,
    \]
    we obtain
    \[
        \Phi(x^{ts}) = \Phi(\Phi^{-1}(\hat{x}^{ts})) = \hat{x}^{ts} = (\hat{x}^t)^s = \left(\widehat{x^t}\right)^s =\widehat{(x^t)^s} = \Phi((x^t)^s), 
    \]
    which implies --- by $\Phi$ being bijective --- that $x^{ts}=(x^t)^s$. Finally, we only prove (ii), because (iii) follows in a similar manner. To do so, let $x,z\in V_+$ and $t\in[0,\infty)$. First, note that 
    \[
        \hat{x}\hat{z} = \Phi(x)\Phi(z) = \Phi(xz) = \widehat{xz},
    \]
    where we used that $\Phi$ is also an algebraic isometry. This then gives us that
    \begin{align*}
        \Phi(x^tz^t) &= \Phi\big(\Phi^{-1}(\hat{x}^t) \Phi^{-1}(\hat{z}^t)\big)=\Phi\big(\Phi^{-1}(\hat{x}^t\hat{z}^t)\big)\\
        &= \hat{x}^t\hat{z}^t= (\hat{x}\hat{z})^t= \widehat{xz}^t \\
        &=\Phi\big(\Phi^{-1}(\widehat{xz}^t)\big) = \Phi((xz)^t). 
    \end{align*}
    Hence, $x^{t} z^{t} = (xz)^t$.
\end{proof}

Finally, we are able to define geometric convexity for maps on AM-algebras, for which we call a subset $A$ of an AM-algebra \textbf{geometrically (or GG-)convex}, if for all $\lambda\in(0,1)$ it holds that $A^{\lambda} \cdot A^{1-\lambda}\subseteq A$, where for two subsets $B,C$ of the same AM-algebra and $\lambda\in(0,1)$ we used that $B\cdot C:=\{b c\mid b\in B,c\in C\}$ (Minkowski product) and $B^{\lambda}:=\{b^\lambda\mid b\in B\}$. Furthermore, we impose the convention that for each $\lambda\in(0,1)$ it holds that $\infty^{\lambda} = \infty$. 

\begin{definition}\label{defi:geometricConvexity}
    Let $V$ be an AM-algebra. For a GG-convex set $E\subseteq V_+$, a map $f:E\rightarrow [0,\infty]^{m}$ is called \textbf{geometrically (or GG-)convex}, if for all $x,y\in E$ and $\lambda\in(0,1)$ it holds that 
    \[
        f\big(x^{\lambda}y^{1-\lambda}\big)\leq f(x)^\lambda f(y)^{1-\lambda}.
    \]
\end{definition}

Next, as for positive homogeneity it is indispensable to discuss the connection of GG-convexity of a map to GG--convexity of the corresponding (geometric) epigraph. 
\begin{proposition}
    Let $V$ be an AM-algebra and assume a map $f:E\rightarrow [0,\infty]^{m}$ on a GG-convex set $E\subseteq V_+$. Consider the following statements:
    \begin{enumerate}
        \item[(a)] $f$ is GG-convex;
        \item[(b)] $\epi(f)$ is GG-convex;
        \item[(c)] $\epiG(f)$ is GG-convex.
    \end{enumerate}
    Then, $(a)\implies (b)\Longleftrightarrow (c)$. For $f:E\rightarrow[0,\infty)^m$ it holds that $(b)\implies (a)$. 
\end{proposition}

\begin{proof}

    First, we prove that (a) implies (b). To do so, let $(x,\alpha),(y,\beta)\in\epi(f)$ and $\lambda\in(0,1)$. Then, it holds that 
    \[
        f(x^{\lambda} y^{1-\lambda})\leq f(x)^{\lambda}  f(y)^{1-\lambda}\leq \alpha^{\lambda}\beta^{1-\lambda}.
    \]
    Hence, $(x,\alpha)^{\lambda}\cdot (y,\beta)^{1-\lambda}=(x^{\lambda}y^{1-\lambda},\alpha^{\lambda}\beta^{1-\lambda})\in\epi(f)$.

    To prove that (b) implies (c), let $(x,\alpha),(y,\beta)\in\epiG(f)$ and $\lambda\in(0,1)$. Then,~$\alpha^{\lambda}\beta^{1-\lambda}\in(0,\infty)^{m}$ and assumption (b) implies that $$(x,\alpha)^\lambda\cdot (y,\beta)^{1-\lambda}=(x^\lambda y^{1-\lambda},\alpha^{\lambda}\beta^{1-\lambda})\in\epi(f)\cap (E\times(0,\infty)^m) = \epiG(f).$$

    Next, we prove that (c) implies (b). In doing so, assume $\lambda\in(0,1)$ and $(x,\alpha),(y,\beta)\in \epi(f)$. If $(x,\alpha),(y,\beta)\in \epiG(f)$ then by assumption $(x,\alpha)^{\lambda}\cdot(y,\beta)^{1-\lambda}\in\epiG(f)\subseteq\epi(f)$. Otherwise, for every $\gamma\in[0,\infty)^m$ we set $I_\gamma:=\{i\in[m]\mid\gamma_i = 0\}$ . Then, for every $i\in I_\alpha \cup I_\beta$ we obtain that
    \[
        f_i(x^{\lambda} y^{1-\lambda})\leq \inf_{\substack{\eta,\kappa\in (0,\infty)^m\\ f(x)\leq \eta,f(y)\leq\kappa}}\eta_i^{\lambda}\,\kappa_i^{1-\lambda} = 0. 
    \]
    This implies that $f_i(x^{\lambda}y^{1-\lambda})=0$ for all $i\in I_\alpha \cup I_\beta$ and hence, $(x,\alpha)^{\lambda}\cdot (y,\beta)^{1-\lambda}\in\epi(f)$.

    Finally, assume $f:E\rightarrow [0,\infty)^m$. To prove that then (b) implies (a), let $x,y\in E$ and $\lambda\in(0,1)$. Then, $(x,f(x)),(y,f(y))\in\epi(f)$, which by GG-convexity of $\epi(f)$ implies that 
    \[
        (x^{\lambda}y^{1-\lambda},f(x)^{\lambda}f(y)^{1-\lambda})\in\epi(f). 
    \]
    Hence, $f(x^{\lambda}y^{1-\lambda})\leq f(x)^{\lambda}f(y)^{1-\lambda}$, which shows that $f$ is GG-convex.
\end{proof}

In contrast to Proposition~\ref{prop:positiveHomogeneity}, we do not obtain that (b) implies (a) when $m=1$, as the next counterexample shows.

\begin{example}
    Let $f:(0,\infty)\rightarrow [0,\infty]$ be given by 
    \[
        f(x) = \begin{cases}
            0, & x = 1,\\
            \infty, & x\neq 1.
        \end{cases}
    \]
    Then, $\epi(f) = \{1\} \times [0,\infty)$. Choose an arbitrary $\lambda\in(0,1)$ and $\alpha,\beta\in[0,\infty)$. Then, we obtain that $(1,\alpha^{\lambda}\beta^{1-\lambda})\in\epi(f)$ and hence, $\epi(f)$ is geometrically convex. However, 
    \[
        f(2^{1-\lambda}) = \infty > 0 = f(1)^{\lambda} f(2)^{1-\lambda}. 
    \]
\end{example}

\subsection{Connection to monetary risk measures}

There is a one-to-one correspondence between RRMs on the interior of the positive cone $\interior_{\lVert.\rVert_{\infty}}(L_+^{\infty})$ and monetary risk measures on $L^{\infty}$. We show that the same correspondence holds for AM-algebras, enabling us to transfer a variety of results from the theory of monetary risk measures to RRMs. 

For this, it is important to note that the interior of the positive cone of an AM-algebra with order unit is always non-empty. The next result shows that this interior refers to an exponential transform of the underlying AM-space. In particular, we show explicitly that the exponential of an element in an AM-space with unit is well-defined.  
\begin{lemma}\label{lem:interiorPositiveCone}
    Let $V$ be an AM-space with unit~$\unit$. Then, for every $x\in V$, $\exp(x)$ exists and lies in $V_+$. Then, set $\exp(V)= \{\exp(x)\mid x\in V\}$. The following statements hold:    
    \begin{enumerate}
        \item[(i)] $\interior_{\normAM{.}}(V_+) = \{x\in V_+\mid \exists r>0: x\geq r\unit\}$;\,\footnote{For a characterization of the interior of $C(K)_+$, see Corollary~\ref{cor:interiorCK}.}
        \item[(ii)] $\interior_{\normAM{.}}(V_+) = \exp(V)$. In particular, $\unit\in \exp(V)$;
        \item[(iii)] For every $y=\exp(x)\in\exp(V)$, $\log(y)$ exists and lies in $V$ with $\log(y) = x$.
    \end{enumerate}
    If $V$ is an AM-algebra with unit $\unit$, then the following holds:
    \begin{enumerate}
        \item[(iv)] For each $x\in V$ and $m\in\R$ we have $\exp(x+m\unit) = \exp(m)\exp(x)$.  
    \end{enumerate}
\end{lemma}

\begin{proof}

    By the Krein-Kakutani Theorem, see~\cite{Kakutani1941} and~\cite[Theorem 9.32]{AliprantisBorder}, $V$ is lattice isometric (unique up to homeomorphism) to $C(K)$ for some compact Hausdorff space $K$ with unit $\mathbf{1}_K$. The lattice isometry is denoted by $\Phi:V\rightarrow C(K)$, or $\hat{x} = \Phi(x)$ with $x\in V$ for short. For each $x\in V$, $\exp(\hat{x})$ is well-defined and is contained in $C(K)_{+}$. So, $\exp(x):=\Phi^{-1}(\exp(\hat{x}))\in V_{+}$ exists for all $x\in V$. Part (iii) follows analogously.

    Now, let us prove (i). First, assume $x\in V_+$ and $r>0$ with $x\geq r\unit$. Then, let $y\in V$ satisfying $\normAM{x-y}<r$. We obtain
    \[
        -r\unit <y-x< r\unit,
    \]
    which in turn yields
    \[
        y = x+(y-x)\geq r \unit -r\unit = 0. 
    \]
    So, $y\in V_+$, which gives us that $B_{\normAM{.}}(x,r)\subseteq V_+$ and hence, $x\in \interior_{\normAM{.}}(V_+)$. 
    
    For the remaining inclusion let $x\in \interior_{\normAM{.}}(V_+)$. Then, there exists $r\in(0,\infty)$ such that  $B_{\normAM{.}}(x,r)\subseteq V_+$. So, it holds that 
    \[
        \left\lVert x-\left(x-\frac{r}{2}\unit\right)\right\rVert_{\unit} = \frac{r}{2}\normAM{\unit} = \frac{r}{2}<r,
    \]
    which implies that 
    \[
        x-\frac{r}{2}\unit\in \ball_{\normAM{.}}(x,r)\subseteq V_+,
    \]
    and hence, $\interior_{\normAM{.}}(V_+)\subseteq \{x\in V_+\mid \exists r>0: x\geq r\unit\}$. 
    
    Next, let us prove (ii). To do so, let $x\in V$. By using the lattice isometry $\Phi$ from the Krein-Kakutani representation, we obtain that $\exp(\hat{x})$ attains a strictly positive minimum $\delta>0$, because it is a real-valued lower semicontinuous function on a compact space, see~\cite[Theorem 2.43]{AliprantisBorder}. Moreover, by~\cite[Theorem 9.17]{AliprantisBorder}, $\Phi^{-1}$ is a positive operator, and by~\cite[Theorem 1.2]{MunozTracedete2025AMalgebras}, we can assume that $\Phi(\unit) = \mathbf{1}_K$, from which we conclude that
    \[
        \exp(x) = \Phi^{-1}(\exp(\hat{x})) \geq \delta \Phi^{-1}(\mathbf{1}_K) = \delta \unit. 
    \]
    Thus, $\exp(x)\in\interior_{\normAM{.}}(V_+)$, which gives us that $\exp(V)\subseteq \interior_{\normAM{.}}(V_+)$. For the inverse inclusion, assume $y\in \interior_{\normAM{.}}(V_+)$. Then, by part (i) we know that $y\geq r\unit $ for some $r>0$. Using again~\cite[Theorem 9.17]{AliprantisBorder}, we have
    \[
        \hat{y} = \Phi(y)\geq r \mathbf{1}_K >0. 
    \]
    This means that the expression $\log(\hat{y})$ is well-defined with $\log(\hat{y})\in C(K)$. Set $x:= \Phi^{-1}(\log(\hat{y}))\in V$, and by $\exp(\hat{x})=\hat{y}$ we get 
    \[
        y = \Phi^{-1}(\hat{y}) = \Phi^{-1}(\exp(\hat{x})) = \exp(x)\in V_+
    \]
    Hence, $\exp(V)\subseteq \interior_{\normAM{.}}(V_+)$.
    
    Then, $\unit\in \exp(V)$ is a direct consequence of (ii) together with~\cite[Theorem 9.40]{AliprantisBorder}.  

    Finally, for part (iv), recall that by~\cite[Theorem 1.2]{MunozTracedete2025AMalgebras}, we can assume that $\Phi(\unit) = \mathbf{1}_K$. Thus, 
    \[
        \Phi(\exp(x+m\unit)) = \exp(m)\exp(\Phi(x))\mathbf{1}_K = \exp(m)\exp(\Phi(x)).
    \]
    This shows that $\exp(x+m\unit) = \exp(m)\exp(x)$. 
\end{proof}

For $n\in\mathbb{N}$, we write $\mathbf{1}_n:=(1,\dots,1)^{\intercal}\in\mathbb{R}^n$, or simply $\mathbf{1}$ if the context is clear. Now we are in a position to introduce monetary risk measures.
\begin{definition}
    Let $V$ be a Riesz space with order unit $\unit$. A map $\riskM:V\rightarrow [-\infty,\infty]^m$ is called a \textbf{(monetary) risk measure} (for $\unit$), if it satisfies the following two properties:
    \begin{enumerate}
        \item[(1)] For all $x,y\in V$ with $x\leq y$ it holds that $\riskM(x)\leq \riskM(y)$ \textit{(monotonicity)};
        \item[(2)] For all $x\in V$ and $m\in\R$ it holds that $\riskM(x+m\unit) = \riskM(x)+m\mathbf{1}$ \textit{(cash-additivity)}.
    \end{enumerate}
    The corresponding acceptance set is $\mathcal{A}_{\rho}:=\{x\in V\mid\rho(x)\leq 0\}$.
\end{definition}

The next result establishes the connection between monetary and return risk measures. Before  stating the result, recall from Lemma~\ref{lem:pushforwardTopology} (i) that if $\tau$ is a topology on an AM-algebra $V$, then $\exp(\tau) := \{\exp(U)\mid U\in\tau\}$ is a topology on $\interior_{\lVert.\rVert_{\unit}}(V_+)$, the so-called pushforward topology.
\begin{proposition}\label{prop:relationToClassicalRM}
    Let $V$ be an AM-algebra with unit $\unit$. Given a map $\rrm:\interior_{\normAM{.}}(V_+)\rightarrow (0,\infty]^{m}$ and set $\riskM = \log\circ \rrm\circ \exp$.\footnote{Note, for $\exp$, we use the construction given in the proof of Lemma~\ref{lem:interiorPositiveCone}.} Then, the following statements hold:
    \begin{enumerate}
        \item[(i)] $\rrm$ is an RRM if and only if $\riskM$ is a risk measure (for $\unit$); 
        \item[(ii)] $\rrm$ is GG-convex if and only if $\riskM$ is convex;
        \item[(iii)] For a topology $\tau$ on $V$, $\rrm$ is $\exp(\tau)$-lower semicontinuous if and only if $\riskM$ is $\tau$-lower semicontinuous; 
        \item[(iv)] For $m=1$, $\rrm$ is proper if and only if $\riskM$ is proper.
    \end{enumerate}
\end{proposition}

\begin{proof} 
    We start by proving (i). The equivalence for monotonicity is an immediate consequence of the monotonicity of $\log$ and $\exp$. Then, let $\rrm$ be positively homogeneous. For all $x\in V$ and $m\in \R$ we obtain by applying Lemma~\ref{lem:interiorPositiveCone} (iv) that
    \[
        \riskM(x+m\unit) = \log\Big(\rrm\big(\exp(x+m\unit)\big)\Big) = \log\Big(\exp(m)\rrm\big(\exp(x)\big)\Big) = \rho(x)+m\mathbf{1}. 
    \]
    Positive homogeneity of $\rrm$ given that $\riskM$ is cash-additive follows analogously by making use of Lemma~\ref{lem:interiorPositiveCone} (iii). To obtain the ``only if'' implication in (ii), let $\lambda\in(0,1)$ and $x,y\in V$. Then, it holds that
    \begin{align*}
        \riskM(\lambda x + (1-\lambda) y) &= (\log \circ \rrm\circ\exp)(\lambda x + (1-\lambda) y)\\ 
        &\leq \log\Big(\rrm\big(\exp(x)\big)^{\lambda}\,\rrm\big(\exp(y)\big)^{1-\lambda}\Big)= \lambda \riskM(x) + (1-\lambda)\riskM(y).
    \end{align*}
    Proving the ``if'' part works analogously. 

    To prove (iii), let $\tilde{\rho}$ be lower semicontinuous with respect to the pushforward topology $\exp(\tau)$ and let $(x_{\alpha})_{\alpha}$ be a net in $V$ with $x_{\alpha}\xrightarrow{\tau} x$ for some $x\in V$. Recall that $\Phi^{-1}\circ\exp\circ \Phi$ is a bijection from~$V$ to $\interior_{\normAM{.}}(V_+)$, where $\Phi$ is the lattice isometry between $V$ and $C(K)$ from the Krein-Kakutani representation, as chosen in the proof of Lemma~\ref{lem:interiorPositiveCone}. Then, by $\exp(\tau) = (\Phi^{-1}\circ\exp\circ \Phi) (\tau)$ and by Lemma~\ref{lem:pushforwardTopology} (ii) we have  $\exp(x_\alpha)\xrightarrow{\exp(\tau)}\exp(x)$. Using the lower semicontinuity of $\rrm$, we then obtain that 
    \[
        \liminf_{\alpha}\rrm(\exp(x_\alpha)) \geq \rrm(\exp(x))
    \]
    and by continuity of the logarithm, we obtain $\liminf_{\alpha}\riskM(x_\alpha) \geq \riskM(x)$. The inverse implication in~(iii) works analogously. To prove (iv), note that $\rrm>0$ implies $\riskM>-\infty$. Moreover, $\dom(\rrm)\neq\emptyset$ is obviously equivalent to $\dom(\riskM)\neq\emptyset$.
\end{proof}

\begin{remark}\label{exam:exponentialTopology}
    If $\tau$ is the topology induced by the norm $\lVert.\rVert_{\unit}$, then  Lemma~\ref{lem:exponentialOfTopology} gives us that $\exp(\tau)=\tau|_{\interior_{\normAM{.}}(V_+)}$. As a consequence of Proposition~\ref{prop:relationToClassicalRM}, additional findings concerning GG-convex functions on the positive real line can be found in Appendix~\ref{sec:supplementaryPositiveRealLine}.
\end{remark}

\section{Continuity and dual representations}\label{sec:continuity}

In the case of $m=1$, the following result shows that an RRM defined on the interior of the positive cone can only attain values in $(0,\infty)$ and it is locally Lipschitz-continuous. 
\begin{theorem}\label{thm:LipschitzContinuity}
    Let $V$ be an AM-algebra with unit $\unit$ and let~${\rrm:\interior_{\normAM{.}}(V_+)\rightarrow (0,\infty]}$ be an RRM. 
    Then, $\rrm\in(0,\infty)$ and $\rrm$ is locally Lipschitz-continuous. 
    
    In particular, if $\rrm$ is GG-convex, then $\rrm$ is lower semicontinuous in the relative topology $\tau|_{\interior_{\normAM{.}}(V_+)}$ with topology $\tau$ being induced by $\normAM{.}$.  
\end{theorem}

\begin{proof}
    Set $\riskM = \log\circ\rrm\circ\exp$. By Proposition~\ref{prop:relationToClassicalRM} (i), $\riskM$ is a risk measure (for $\unit$). Hence, we can use the cash-additivity of $\riskM$ to represent it as follows:
    \[
        \riskM(x) = \inf\{m\in\R\mid X-m\unit\in\mathcal{A}_{\riskM}\}.
    \]
    By $\unit$ being an interior point of $V_+$, we can apply~\cite[Proposition 3.1]{Farkas2014BeyondCashAdditiveRMs} to obtain that $\rho$ is finite-valued and Lipschitz-continuous (with Lipschitz constant $L>0$). Hence, $\rrm<\infty$, by Proposition~\ref{prop:relationToClassicalRM} (iv), and therefore, $\rrm\in(0,\infty)$. To prove that $\rrm$ is locally Lipschitz-continuous, choose an arbitrary $x\in\interior_{\normAM{.}}(V_+)$ with $x\in\left[\varepsilon \unit, c\unit\right]$ for some $\varepsilon,c>0$. Then, as neighborhood of $x$ choose $\ball_{\normAM{.}}\big(x,\frac{\varepsilon}{2}\big)\subseteq \interior_{\normAM{.}}(V_+)$. Using again the lattice isometry of the Krein-Kakutani representation between $V$ and $C(K)$ for the corresponding compact Hausdorff space $K$, for every $y\in \ball_{\normAM{.}}\big(x,\frac{\varepsilon}{2}\big)$ we have $\log(y)\in\left[\log(\frac{\varepsilon}{2}\unit), \log((\frac{\varepsilon}{2}+c)\unit)\right]$. Thus, by the fact that $\log^{\prime}(t) = \frac{1}{t}\leq \frac{2}{\varepsilon}$ on $\left[\frac{\varepsilon}{2},\infty\right)$ ---  i.e.,~$|\log(u)-\log(s)|\leq \frac{2}{\varepsilon}|u-s|$ for all $u,s\in\left[\frac{\varepsilon}{2},\infty\right)$ --- the Lipschitz-continuity of $\riskM$ implies for all $y\in \ball_{\normAM{.}}\big(x,\frac{\varepsilon}{2}\big)$ that
    \[
        \normAM{\rho(\log(y))-\riskM(\log(x))}\leq L\normAM{\log(y)-\log(x)}\leq \frac{2L}{\varepsilon}\normAM{y -x}.
    \]
    So, $\rho\circ\log$ is Lipschitz-continuous on $\ball_{\normAM{.}}\big(x,\frac{\varepsilon}{2}\big)$. Furthermore, $\rho\circ\log$ is bounded on $\ball_{\normAM{.}}\big(x,\frac{\varepsilon}{2}\big)$. Indeed, for all $y\in \ball_{\normAM{.}}\big(x,\frac{\varepsilon}{2}\big)$ we have
    \begin{align*}
        \normAM{\riskM(\log(y))}&\leq \normAM{\riskM(\log(y))-\riskM(\log(x))}+\normAM{\riskM(\log(x))}\\
        &\leq \frac{2L}{\varepsilon}\normAM{ y-x}+\normAM{\riskM(\log(x))}\\
        &\leq L+\normAM{\riskM(\log(x))}\in\R.
    \end{align*}
    Thus, by the fact that the exponential function is also Lipschtz-continuous on bounded sets, we obtain that the composition $\rrm=\exp\circ\riskM\circ \log$ is Lipschitz-continuous on $\ball_{\normAM{.}}\big(x,\frac{\varepsilon}{2}\big)$ and hence, $\rrm$ is locally Lipschitz-continuous.

    If $\rrm$ is GG-convex, then $\riskM$ is convex. Recall, Lipschitz-continuity of $\riskM$ implies continuity of $\riskM$ with respect to the norm topology $\tau$. Moreover, $\tau$ and $\sigma(V,V^{\prime})$ are consistent to the dual pair~$\langle V,V^{\prime}\rangle$. By~\cite[Corollary 5.99]{AliprantisBorder}, the set of lower semicontinuous convex maps is the same in all topologies consistent with respect to a fixed dual pair. Hence, $\riskM$ is $\sigma(V,V^{\prime})$-lower semicontinuous. The claim then follows by applying $\rrm = \exp\circ\riskM\circ\log$ and Lemma~\ref{lem:exponentialOfTopology}.
\end{proof}

The following example shows that $\rrm$ is not globally Lipschitz-continuous in general.
\begin{example}
    We make use of the certainty equivalent for logarithmic utility. To do so, we choose $V = L^{\infty}(\Omega;\R)$ and $\riskM(X) = \E[X]$. Then, for a set $A\in \mathcal{F}$ with $\P(A) = \frac{1}{2}$ set $X_n = \frac{1}{n} \cdot\mathbf{1}_A + \mathbf{1}_{A^{\complement}}$ as well as $Y_n = \frac{2}{n} \cdot\mathbf{1}_A + \mathbf{1}_{A^{\complement}}$. Then, 
    \begin{align*}
        \rrm(X_n) &= (\exp\circ\riskM\circ\log)(X_n) = \exp\Big(-\frac{1}{2}\log(n)\Big) = \sqrt{\frac{1}{n}},\\ 
        \rrm(Y_n) &= (\exp\circ\riskM\circ\log)(Y_n) = \exp\Big(\frac{1}{2}\log\Big(\frac{2}{n}\Big)\Big) = \sqrt{\frac{2}{n}}. 
    \end{align*}
    Together with $\lVert X_n-Y_n\rVert_{\infty} = \frac{2}{n}-\frac{1}{n} = \frac{1}{n}$ we obtain
    \[
        \lim\limits_{n\rightarrow\infty} \frac{|\rrm(X_n)-\rrm(Y_n)|}{\lVert X_n-Y_n\rVert_{\infty}} =  \lim\limits_{n\rightarrow\infty} \frac{(\sqrt{2}-1)n^{-\frac{1}{2}}}{1/n} = (\sqrt{2}-1)\lim\limits_{n\rightarrow\infty}\sqrt{n} = \infty. 
    \]
    Hence, no Lipschitz-constant exists, and therefore, $\rrm$ is not globally Lipschitz-continuous.
\end{example}

By exploiting the connection to monetary risk measures, we derive dual representations for RRMs. Dual representations play a central role for portfolio optimization with risk measures, see e.g.,~the considerations in~\cite{ruszczynskiShapiro2006} and~\cite{Mastrogiacomo2015}. Also they have a central role in proving many remarkable results concerning monetary risk measures. For instance, see the proofs of~\cite[Theorem 3.3]{AraratFeinstein2024Separability} and~\cite[Theorem 3.5]{Liebrich2019RiskSharing}. Dual representations can also be the starting point to define risk measures, see e.g.,~the approach in~\cite[Section 4]{Liebrich2017ModelSpaces} or the definition of systemic optimal allocations via dual representations in~\cite{Biagini2020SystemicRisk}. However, only a few dual representations for RRMs are available and all rely on domains of random variables, see e.g.,~\cite{laudage2025MultiAssetRRM}. Therefore, we extend these results to the interior of the positive cone of an AM-algebra as domain.

To do so, if $V$ is a Riesz space and $E\subseteq V$, then $E^{\disjointComplement} := \{x\in V\mid \forall y\in E: |x|\wedge|y| = 0\}$ is the \textbf{disjoint complement} of $E$. The order dual and the order continuous dual of $V$ are denoted by  $V^{\sim}$ and~$V_n^{\sim}$, respectively; see~\cite[Section 8.10]{AliprantisBorder} for additional information on these terms. For brevity, the linear space of singular functionals is denoted by $V_s^{\sim} := (V_n^{\sim})^{\disjointComplement}$. For a \textcolor{blue}{normed Riesz space} $V$, we define $V_{+}^{\prime}:=\{\dualElementClassic\in V^{\prime}\mid\forall x\in V_+: \dualElementClassic(x)\geq 0\}$. 

Let $V$ be an AM-algebra with unit. Our dual representations rely on the following geometric transformations of the sets $(V_n^{\sim})_{+}$ and $(V_s^{\sim})_{+}$:
\begin{align*}
    V_{n,+}^{\text{GG,}\sim}&:=\big\{\dualElementGG_c:\interior_{\lVert.\rVert_{\infty}}(V_+)\rightarrow (0,\infty)\,\big|\, \log\circ \dualElementGG_c\circ \exp\in (V^{\sim}_n)_{+}\big\},\\
    \text{and}\quad V_{s,+}^{\text{GG,}\sim}&:=\big\{\dualElementGG_s:\interior_{\lVert.\rVert_{\infty}}(V_+)\rightarrow (0,\infty)\,\big|\, \log\circ \dualElementGG_s\circ \exp\in (V^{\sim}_s)_{+}\big\}. 
\end{align*}
Note that a functional $l$ contained in one of those sets is positively homogeneous and satisfies $l(x^{a}) = l(x)^{a}$ for all $x\in\interior_{\lVert.\rVert_{\infty}}(V_+)$ and $a\in(0,\infty)$, i.e.,~$l$ is geometrically linear.

\begin{theorem}\label{thm:dualRepresentation}
    Let $V$ be an infinite-dimensional AM-algebra with unit $\unit$ 
    and assume a proper map $\rrm:\interior_{\lVert.\rVert_{\infty}}(V_+)\rightarrow (0,\infty]$. Then, $\rrm$ is a GG-convex RRM if and only if it admits the representation 
    \begin{align}\label{eq:dualRRM_interiorCone}
        \rrm(x) = \max\limits_{\dualElementGG\in L}\frac{\dualElementGG(x)}{\tilde{a}\big(\dualElementGG\big)}, \quad x\in \interior_{\lVert.\rVert_{\infty}}(V_+),
    \end{align}
    where 
    \begin{align}\label{eq:setOptimizationDualRepresentation}
        L = \Bigg\{\dualElementGG_c\cdot \dualElementGG_s\,\Bigg|\, \dualElementGG_c\in V_{n,+}^{\text{GG,}\sim}, \dualElementGG_s\in V_{s,+}^{\text{GG,}\sim},\sup\limits_{\substack{x\geq 1\\ \lVert \log(x)\rVert_{\infty}=1}} \Big(\dualElementGG_c(x)\dualElementGG_s(x)\Big) = e\Bigg\},
    \end{align}

    \noindent and $\tilde{a}:L\rightarrow (0,\infty]$ is given by
    \begin{align}\label{eq:penaltyFunction}
        \tilde{a}\big(\dualElementGG\big) = \sup\limits_{x\in \interior_{\lVert.\rVert_{\infty}}(V_+)}\frac{\dualElementGG(x)}{\rrm(x)} = \sup\limits_{x\in \mathcal{B}_{\rrm}}\dualElementGG(x), \quad \dualElementGG\in L.
    \end{align}
\end{theorem}

Before proving this result, let us discuss its implications. Equation~\eqref{eq:dualRRM_interiorCone} shows that we have a multiplicative connection between $l$ and $\tilde{a}$. In contrast, for monetary risk measures this is a summation. Moreover, the set of dual elements $L$ in~\eqref{eq:setOptimizationDualRepresentation} only contains elements satisfying the normalization condition of $\sup\limits_{\substack{x\geq 1,\\ \lVert \log(x)\rVert_{\infty}=1}} \Big(\dualElementGG_c(x)\dualElementGG_s(x)\Big) = e$. This highlights the impact of the exponential transform by leading to normalization with respect to $e$, instead of $1$ as it is the case for monetary risk measures.

The proof now relies on $V^{\prime} = V^{\sim}$. Therefore, let us collect some facts about the dual space $V^{\prime}$ of an AM-space $V$, which implicitly also occurs in the definition of the set $L$ in Theorem~\ref{thm:dualRepresentation}, because $V^{\prime} = (V^{\sim}_n)\oplus(V^{\sim}_s)$. First, $V^{\prime}$ is an AL-space, see~\cite[Theorem 9.27]{AliprantisBorder}. Second, $V^{\prime}$ is lattice isometric to an $L^{1}(\mu)$-space for some measure $\mu$, see~\cite[Theorem 12.26]{AliprantisBurkinshaw}. Finally, $V^{\prime} = (V^{\prime\prime})^{\sim}_n$, i.e.,~$V^{\prime}$ is equal to the order continuous dual of $V^{\prime\prime}$, see~\cite[Theorem 9.34]{AliprantisBorder}.

\begin{proof}[Proof of Theorem~\ref{thm:dualRepresentation}]
    We first prove the ``only if'' part. By Proposition~\ref{prop:relationToClassicalRM} and the proof of Theorem~\ref{thm:LipschitzContinuity}, we have that $\riskM=\exp\circ\rrm\circ\log$ is a proper, convex and Lipschitz-continuous risk measure. Hence, $\riskM$ is also $\sigma(V,V^{\prime})$-(lower semi-)continuous. By~\cite[Theorem 9.38]{AliprantisBorder} and $V$ being infinite-dimensional, we obtain that $V$ is non-reflexive. By $V$ being an AM-algebra it is also a Banach space. Concluding, $V$ is a non-reflexive Banach space, which allows for the application of~\cite[Theorem 4.2]{Kountzakis2011RM_nonReflexiveBanachSpaces}, from which we obtain the dual representation
    \begin{align}\label{eq:proofDualRepresentation2}
        \riskM(x) = \sup_{\dualElementClassic\in B_{\unit}}(\dualElementClassic(x)-\riskM^{\star}(\dualElementClassic)), 
    \end{align}
    where $\riskM^{\star}(\dualElementClassic) = \sup\limits_{x\in V}\big(\dualElementClassic(x)-\rho(x)\big)$ for $\dualElementClassic\in V^{\prime}$ is the Fenchel-Legendre transform of $\rho$ and $B_{\unit} = \left\{\dualElementClassic\in V^{\prime}_{+}\mid \hat{\unit}(\dualElementClassic) = 1\right\}$ with $\hat{\unit}$ being the embedding of $\unit$ in $V^{\prime\prime}$, i.e.,~$\hat{\unit}(\dualElementClassic) = \dualElementClassic(\unit)$ for all $\dualElementClassic\in V^{\prime}$. Next, note that cash-additivity of $\riskM$ implies 
    \[
        \riskM(x) = \inf\{m\in \R\,|\,x-m\unit\in \mathcal{A}_{\riskM}\},\quad x\in V.
    \]
    Using a standard argumentation\footnote{The ansatz to show that the Fenchel-Legendre conjugate $\riskM^{\star}$ of the risk measure $\riskM$ and the (upper) support function~$\sigma_{\mathcal{A}_{\riskM}}$ of the acceptance set $\mathcal{A}_{\rho}$ are equal can be found in the proof of~\cite[Proposition 3.9]{FrittelliScandolo}.} we obtain for all $\dualElementClassic\in V^{\prime}$ that 
    \begin{align*}
        \riskM^{\star}(\dualElementClassic) &= \sup\limits_{x\in V}\Big(\dualElementClassic(x)-\inf\{m\in\R\,|\, x-m\unit\in \mathcal{A}_{\riskM}\}\Big)\\
        &= \sup_{y\in \mathcal{A}_{\rho}} \dualElementClassic(y) +\sup_{m\in\R}\big(\dualElementClassic(\unit)-1\big)m\\
        &= \sigma_{\mathcal{A}_{\riskM}}(\dualElementClassic) +\infty\cdot (1- \mathbf{1}_{\{1\}}\big(\dualElementClassic(\unit)\big)\big).
    \end{align*}
    In particular, if $\dualElementClassic\in B_{\unit}$, then $\riskM^{\star}(\dualElementClassic) = \sigma_{\mathcal{A}_{\riskM}}(\dualElementClassic)<\infty$. Now, we show that the maximum is attained. To do so, let $x\in V$. By $\rho$ being $\sigma(V,V^{\prime})$-continuous,~\cite[Theorem 7.12]{AliprantisBorder} states that there exists $\dualElementClassic^{\star}\in V^{\prime}$ such that 
    \[
        \forall y\in V:\dualElementClassic^{\star}(x)-\riskM(x)\geq \dualElementClassic^{\star}(y)-\riskM(y)
    \]
    and consequently, we obtain $\rho(x) = \dualElementClassic^{\star}(x) -\riskM^{\star}(\dualElementClassic^{\star})$. $\dualElementClassic^{\star}\in B_{\unit}$ follows by a standard argumentation, for which we refer to the proof of~\cite[Theorem 4.16]{FoeSch}. 
    
    Then, by putting the dual representation~\eqref{eq:proofDualRepresentation2} into $\rrm = \log\circ\riskM\circ\exp$ and rewrite it, we obtain that
    \begin{align}\label{eq:proofDualRepresentation1}
        \rrm(x) = \sup\limits_{\dualElementClassic\in B_{\unit}}\frac{(\exp\circ \dualElementClassic\circ\log)(x)}{(\exp\circ \riskM^{\star})(\dualElementClassic)}. 
    \end{align}

    Then we rewrite the defining condition $\hat{\unit}(\dualElementClassic) = 1$ of the set $B_{\unit}$. To do so, we obtain by~\cite[Theorem 9.31]{AliprantisBorder} that $\hat{\unit}$ refers to the unit of the original AM-space $V$. Hence, we can apply~\cite[Theorem 9.30]{AliprantisBorder}, which gives us $\hat{\unit}(\dualElementClassic) = \lVert \dualElementClassic^+\rVert_{V^{\prime}} - \lVert \dualElementClassic^-\rVert_{V^{\prime}}$, where 
    \[
        \lVert \dualElementClassic\rVert_{V^{\prime}} = \sup\limits_{\substack{x\in V\\ \lVert x\rVert = 1}}|\dualElementClassic(x)|,\quad \dualElementClassic\in V^{\prime},
    \]  
    denotes the dual norm. First, we show that $\lVert \dualElementClassic^-\rVert_{V^{\prime}} = 0$. To do so, by applying the Riesz-Kantorovi\v{c} formula, see~\cite[Theorem 1.13]{AliprantisBurkinshaw}, we obtain for all $x\geq 0$ that
    \[
        \dualElementClassic^-(x) = \sup_{0\leq y\leq x}\big(-\dualElementClassic(y)\big)\leq 0,
    \]
    where the inequality follows from $\dualElementClassic\in V^{\prime}_+$. Hence, $\dualElementClassic^-|_{V_+} = 0$. So, for $x\in V$ it holds that 
    \[
        \dualElementClassic^-(x) = \dualElementClassic^-(x^{+}) - \dualElementClassic^{-}(x^{-}) = 0-0 = 0.
    \]
    Thus, $\lVert \dualElementClassic^{-}\rVert_{V^{\prime}} = 0$. Finally, by $\dualElementClassic^{+}$ being a positive operator between Riesz spaces, we obtain that $|\dualElementClassic(x)|\leq \dualElementClassic(|x|)$ holds for all $x\in V$. Hence,
    \[
        \sup\limits_{\substack{x\in V\\ \lVert x\rVert = 1}}|\dualElementClassic^{+}(x)| \leq \sup\limits_{\substack{x\in V\\ \lVert x\rVert = 1}}\dualElementClassic^{+}(|x|) = \sup\limits_{\substack{x\geq 0\\ \lVert x\rVert = 1}}\dualElementClassic^{+}(x)=\sup\limits_{\substack{x\geq 0\\ \lVert x\rVert = 1}}\dualElementClassic(x). 
    \]

    Next, by $V$ being a Fr\'{e}chet lattice (under its norm topology) and by~\cite[Theorems 8.28 and 9.11]{AliprantisBorder} we obtain that 
    \[
        V^{\prime} = V^{\sim} =  V^{\sim}_n\oplus(V^{\sim}_n)^{\disjointComplement} = V^{\sim}_n\oplus V^{\sim}_s.
    \]
    Then, $V^{\prime}_+ = (V^{\sim}_n)_+\oplus(V^{\sim}_s)_+$ follows by Lemma~\ref{lem:bandDecomposition}. Hence, 
    \[
        B_{\unit} = \left\{\dualElementClassic_c+\dualElementClassic_s\,\middle|\, \dualElementClassic_c\in (V^{\sim}_n)_{+},\dualElementClassic_s\in (V^{\sim}_s)_{+},\sup\limits_{\substack{x\geq 0\\ \lVert x\rVert_{\infty}=1}} \Big(\dualElementClassic_c(x) + \dualElementClassic_s(x)\Big) = 1\right\}. 
    \]
    The statement follows by rewriting~\eqref{eq:proofDualRepresentation1} and $\sigma_{\mathcal{A}_{\riskM}}$ based on the transforms $\dualElementGG = \exp\circ \dualElementClassic \circ \log$ and $\tilde{a} = \exp\circ \riskM^{\star}\circ\log$ as well as by setting $L = \Big\{\exp\circ \dualElementClassic\circ\log\, \Big|\, \dualElementClassic\in B_{\unit}\Big\}$. 

    To prove the ``if'' part, assume that the representation~\eqref{eq:dualRRM_interiorCone} holds. Then, by the transform $\riskM=\exp\circ\rrm\circ\log$ and applying again~\cite[Theorem 4.2]{Kountzakis2011RM_nonReflexiveBanachSpaces} we obtain that $\riskM$ is a convex risk measure and hence, by Proposition~\ref{prop:relationToClassicalRM} we obtain that $\rrm$ is a GG-convex RRM. 
\end{proof}

\begin{remark}
    The proof of Theorem~\ref{thm:dualRepresentation} relies on the dual representation for monetary risk measures on non-reflexive Banach spaces in~\cite{Kountzakis2011RM_nonReflexiveBanachSpaces}, which does not ensure that the supremum is attained. We established this by following the hint in~\cite[Remark~18]{Farkas2015MultipleEligible}. We could also use other dual representations, provided the domain covers AM-spaces with unit, e.g.,~\cite[Theorem 3]{Farkas2015MultipleEligible}. However, their framework includes multiple eligible assets, which makes the treatment more involved. We also do not apply~\cite[Theorem 1]{BiaginiFrittelli}, since $\riskM$ is not assumed to be cash-additive, so the condition $\hat{\unit}(l) =1$ does not appear in their dual representation. Moreover, the results in~\cite[Section 3]{Konstantinides2011RM_nonEmptyConeInterior} are not applicable, as an AM-space cannot be both reflexive and infinite dimensional, see~\cite[Theorem 9.38]{AliprantisBorder}.
\end{remark}

For illustration, let us specify the order duals from Theorem~\ref{thm:dualRepresentation} for the case of $V = L^{\infty}(\Omega;\R^n)$. For this, we define $\ba_{\mathbb{P}} := \{\mu\in\ba\,|\, \mu\ll\mathbb{P}\}$, where $\ba$ denotes the linear space of all finitely additive measures on $(\Omega,\mathcal{F})$ whose total variation is finite. 

\begin{example}\label{exam:YosidaHewittDecomposition}
    If $V= L^{\infty}(\Omega;\R^n)$, then $V^{\prime} = \big(\ba_{\mathbb{P}}\big)^n$, and by the Yosida-Hewitt decomposition of~$\ba_{\mathbb{P}}$, see~\cite{YosidaHewitt1952}, we obtain that $V^{\prime} = L^1(\Omega;\R^n)\oplus L^1(\Omega;\R^n)^{\disjointComplement}$. Note, $L^1(\Omega;\R^n)^{\disjointComplement}$
    refers to the set of purely additive measures absolutely continuous with respect to $\mathbb{P}$. The fact that the linear spaces in the Yosida-Hewitt decomposition are disjoint complements allows us to apply Lemma~\ref{lem:bandDecomposition} to obtain $V_+^{\prime}=L^1(\Omega;\R^n)_+\oplus (L^1(\Omega;\R^n)^{\disjointComplement})_+$.

    The connection to the dual representation of monetary risk measures in~\cite[Theorem 4.16]{FoeSch} becomes apparent by noting that, for each $l\in L$ there exists $\mu\in \big(\ba_{\mathbb{P}}\big)^n$ with $\sum_{i=1}^n \mu_i(\Omega) = 1$ and $l(x) = \exp\Big(\sum_{i=1}^n \int \log(x_i)\diff \mu_i\Big)$ and hence, 
    \[
        \tilde{a}(l) = \sup\limits_{x\in \mathcal{B}_{\rrm}}\exp\Big(\sum_{i=1}^n \int \log(x_i)\diff \mu_i\Big).
    \]
\end{example}

The next natural step is to determine conditions under which singular functionals, i.e.,~$\dualElementGG_s\in V_{s,+}^{\text{GG,}\sim}$, no longer appear in the dual representation. This is possible under further continuity assumptions on the RRM. However, the price to pay is that the supremum in the resulting dual representation does not need to be attained. 

\begin{lemma}\label{lem:weakLowerSemicontinuity}
    Let $V$ be an AM-space with unit. The map $\rrm:\interior_{\lVert.\rVert_{\infty}}(V_+)\rightarrow (0,\infty]$ is an RRM and~$\riskM = \log\circ\rrm\circ \exp$. The following statements are equivalent:
    \begin{enumerate}
        \item[(a)] $\riskM$ is $\sigma(V,V_n^{\sim})$-lower semicontinuous;
        \item[(b)] $\rrm$ is $\exp\big(\sigma(V,V_n^{\sim})\big)$-lower semicontinuous.
    \end{enumerate}
\end{lemma}

\begin{proof}
    $\riskM$ being $\sigma\big(V,V_n^{\sim}\big)$-lower semicontinuous is equivalent to $\{x\in V\mid\rho(x)\leq c\}$ being $\sigma\big(V,V_n^{\sim}\big)$-closed for all $c\in\mathbb{R}$. By $\riskM = \log\circ\rrm\circ \exp$, this is equivalent to $\big\{x\in \interior_{\lVert.\rVert_{\infty}}(V_+)\,\big|\, \rrm(x)\leq \tilde{c}\big\}$ being $\exp\big(\sigma(V,V_n^{\sim})\big)$-closed for all $\tilde{c}\in(0,\infty)$. By $\rrm>0$,  for all $\tilde{c}\in(-\infty,0]$ we have 
    \[
        \big\{x\in \interior_{\lVert.\rVert_{\infty}}(V_+)\,\big|\, \rrm(x)\leq \tilde{c}\big\} = \emptyset,
    \]
    and hence, those sets are $\exp\big(\sigma(V,V_n^{\sim})\big)$-closed. This proves the claim.
\end{proof}

\begin{remark}
    In suitable settings, the $\sigma(V,V_n^{\sim})$-lower semicontinuity of a monetary risk measure $\riskM$ is equivalent to the Fatou property, see e.g.,~\cite[Theorem 4.33]{FoeSch} for the case of $V=L^{\infty}$. As mentioned in~\cite{GaoXanthos}, for Banach lattices, one needs to adjust the definition of the Fatou property by substituting almost sure convergence by unbounded order convergence. For an introduction to unbounded order convergence, we refer to~\cite{Gao2014,GaoXanthos2014}.
\end{remark}

For brevity, we keep the proof of the final dual representation as short as possible. 

\begin{corollary}\label{cor:dualRepresentation}
    Let $V$ be an AM-algebra with unit $\unit$ and let $\rrm:\interior_{\lVert.\rVert_{\infty}}(V_+)\rightarrow (0,\infty]$ be a proper, GG-convex and $\exp\big(\sigma(V,V^{\sim}_n)\big)$-lower semicontinuous map. Then, $\rrm$ admits the dual representation 
    \begin{align*}\label{eq:dualRRM_interiorCone2}
        \rrm(x) = \sup\limits_{\dualElementGG\in L}\frac{\dualElementGG(x)}{\tilde{a}\big(\dualElementGG\big)}, \quad x\in \interior_{\lVert.\rVert_{\infty}}(V_+),
    \end{align*}
    where 
    \[
        L = \Bigg\{\dualElementGG\in V_{n,+}^{\text{GG,}\sim}\,\Bigg|\, \sup\limits_{\substack{x\geq 1\\ \lVert \log(x)\rVert_{\infty}=1}} \dualElementGG(x) = e\Bigg\},
    \]
    and $\tilde{a}:L\rightarrow (-\infty,\infty]$ is given as in~\eqref{eq:penaltyFunction}.
\end{corollary}

\begin{proof}
    Use the transform $\riskM = \log\circ\rrm\circ\exp$, the imposed assumptions and Lemma~\ref{lem:weakLowerSemicontinuity} to obtain that  $\riskM$ is finite-valued, convex and $\sigma(V,V^{\sim}_n)$-lower semicontinuous. Thus, by applying the Fenchel-Moreau theorem~\cite[Theorem 2.3.3]{Zalinescu} and using analogous arguments as in the proof of Theorem~\ref{thm:dualRepresentation}, the desired claim follows.  
\end{proof}

\section{Systemic return risk measures}\label{sec:multivariateRRM}

Now, we represent RRMs for multivariate losses as compositions of RRMs for one-dimensional losses  and aggregation functions. Our decompositions are inspired by the theory of systemic risk measures. For simplicity, we impose the following assumption throughout the rest of this section.
\begin{assumption}
    Let $I\subseteq (0,\infty)$ be an interval. Moreover, let $V = L^{\infty}(\Omega;\R^n)$ and $E\subseteq L^{\infty}_{++}$ being a GG-convex cone with $(0,\infty)\subseteq E$, i.e.,~$E$ contains all positive constants. The set 
    \[
        E_I:=\{X\in E\mid X(\omega)\in I\ \omega\text{-a.s.}\}
    \]
    is assumed to be a GG-convex cone.
\end{assumption}

\begin{remark}
    The important cases for $I$ are $(0,\infty)$ and $[1,\infty)$. If $I=(0,\infty)$, then $E_I = E$. If $I =[1,\infty)$, then $E_I$ contains all random variables $X\in E$ with $X\geq 1$ a.s., i.e.,~$E_I$ refers to random variables leading to pure logarithmic losses ($\log(X)\geq 0$ a.s.). 
\end{remark}

Note that also $E^n$ is a GG-convex cone. For our result, we adapt the definition of systemic risk measures to the framework of RRMs.
\begin{definition}\label{def:systemicRRM}
    An increasing map $\rrm:E^n\rightarrow (0,\infty)$ is a \textbf{systemic RRM}, if it satisfies the following properties:
    \begin{enumerate}
        \item[(1)] $\rrm\big((0,\infty)^n\big) = I$ \textit{($I$-surjectivity on constants)};
        \item[(2)] $\rrm|_{(0,\infty)^n}(E^n) = E_I$;
        \item[(3)] For all $\mathbf{X},\mathbf{Y}\in E^n$ with $\rrm(\mathbf{X}(\omega))\leq \rrm(\mathbf{Y}(\omega))$ $\omega$-a.s.~it holds that $\rrm(\mathbf{X})\leq \rrm(\mathbf{Y})$ \textit{(preference consistency)}.
    \end{enumerate}
    A systemic RRM $\rrm$ is a \textbf{GG-convex systemic RRM}, if it is GG-convex and satisfies the following:
    \begin{enumerate}
        \item[(4)] For all $\mathbf{X},\mathbf{Y},\mathbf{Z}\in E^n$ and $\lambda\in(0,1)$ with $\rrm(\mathbf{Z}(\omega)) =  \rrm(\mathbf{X}(\omega))^{\lambda}\rrm(\mathbf{Y}(\omega))^{1-\lambda}$ $\omega$-a.s.~it holds that $\rrm(\mathbf{Z})\leq \rrm(\mathbf{X})^{\lambda}\rrm(\mathbf{Y})^{1-\lambda}$ \textit{(risk GG-convexity)}. 
    \end{enumerate} 
    If $\rrm$ is a positively homogeneous GG-convex systemic RRM, then we call it a \textbf{GG-coherent systemic RRM}.\footnote{In contrast to classical works, we omit the assumption of being normalized. This is not problematic, as pointed out by~\cite{chen2013SystemicRiskMeasures}, normalization ``[\dots] \textit{is simply a convenient choice of scaling and is imposed without loss of generality.}''} 
\end{definition}

The properties of being increasing, GG-convex and positively homogeneous are the standard conditions of RRMs, which we also used in the previous sections. The surjectivity properties (1) and (2) ensure surjectivity of functions which we create out of $\rrm$ in the proof of Theorem~\ref{thm:representationSystemicRRM} below. Properties (3) and (4) are the most crucial ones, but they are essential for our result as well. They are based on comparing the risks of the constant losses for specific events $\omega$. Analogous assumptions are standard in the literature on systemic risk measures and they were first proposed in~\cite{chen2013SystemicRiskMeasures}. After presenting Theorem~\ref{thm:representationSystemicRRM}, we demonstrate situations in which either condition (3) or condition (4) fails and discuss corresponding consequences in Examples~\ref{exam:notPreferenceConsistent} and~\ref{exam:notRiskGGconvex}, thereby addressing an important gap in the literature. 

Next, we introduce the definition of a return aggregation function.
\begin{definition}
    An increasing and surjective map $\aggregation:(0,\infty)^n\rightarrow I$ with $\aggregation(E^n) = E_I$ is called a \textbf{return aggregation function}. If $\aggregation$ is in addition GG-convex, then we call it a \textbf{GG-convex return aggregation function}. If $\aggregation$ is a positively homogeneous GG-convex return aggregation function, then we call it a \textbf{GG-coherent return aggregation function}.
\end{definition}

We are now able to represent a systemic RRM as the composition of a single map with domain~$E$ and a return aggregation function. 
\begin{theorem}\label{thm:representationSystemicRRM}
    A map $\rrm:E^n\rightarrow (0,\infty)$ is a (GG-convex, resp.~GG-coherent) systemic RRM if and only if there exist maps $\rrm_0:E_I\rightarrow (0,\infty)$ and $\aggregation:(0,\infty)^n\rightarrow I$ such that 
    \[ 
        \rrm = \rrm_0 \circ\aggregation
    \]
    and for which the following hold:
    \begin{enumerate}
        \item[(i)] $\aggregation$ is a (GG-convex, resp.~GG-coherent) return aggregation function; 
        \item[(ii)] $\rrm_0$ is increasing (and GG-convex, resp.~GG-convex and positively homogeneous) and satisfies the constancy property on $I$, i.e.,~$\rrm_0(x) = x$ for all $x\in I$.
    \end{enumerate}
\end{theorem}

\begin{proof}
    To prove the ``only if'' part, assume that $\rrm$ is a systemic RRM. Then, we set
    \[
        \aggregation(x) := \rrm(x),\quad x\in (0,\infty)^n.
    \]
    We have to show that $\aggregation$ is a return aggregation function. By $\rrm$ being increasing, we obtain that~$\aggregation$ is increasing. Then, $I$-surjectivity on constants of $\rrm$ gives us that
    \[
        \aggregation\big((0,\infty)^n\big) = \rrm\big((0,\infty)^n\big) = I,
    \]
    i.e.,~$\aggregation:(0,\infty)^n\rightarrow I$ is surjective. By $\rrm|_{(0,\infty)^n}(E^n) = E_I$ we have
    \[
        \aggregation(E^n) = \rrm|_{(0,\infty)^n}(E^n) = E_I. 
    \]
    In total, $\aggregation$ is a return aggregation function. Now, assume that $\rrm$ is GG-convex and let $x,y\in(0,\infty)$ and~$\lambda\in(0,1)$. Then, we get
    \[
        \aggregation\big(x^{\lambda}y^{1-\lambda}\big) = \rrm\big(x^{\lambda}y^{1-\lambda}\big)\leq \rrm(x)^{\lambda}\rrm(y)^{1-\lambda} = \aggregation(x)^{\lambda}\aggregation(y)^{1-\lambda}. 
    \]
    Thus, $\aggregation$ is GG-convex. If $\rrm$ is positively homogeneous, then positive homogeneity of $\aggregation$ follows analogously.

    Next, we define the map $\rrm_0:\aggregation(E^n) = E_I\rightarrow(0,\infty)$ for each $X\in E_I$ as
    \begin{align}\label{eq:proofSystemicDefinitionSingleRM}
        \rrm_0(X):= \rrm(\mathbf{X}),
    \end{align}
    where $\mathbf{X}\in E^n$ with $\aggregation(\mathbf{X}) = X$ almost surely. First, we need to prove that this definition is well-defined. To do so, choose an arbitrary $X\in E_I$ and let $\mathbf{X},\mathbf{Y}\in E^n$ with 
    \[
        X = \aggregation(\mathbf{X}) = \aggregation(\mathbf{Y})\quad \text{a.s.}
    \]
    Then, there exists $N\in\mathcal{F}$ with $\P(N) = 0$ such that 
    \[
        \rrm\big(\mathbf{X}(\omega)\big) = \aggregation(\mathbf{X})(\omega) = \aggregation(\mathbf{Y})(\omega) = \rrm\big(\mathbf{Y}(\omega)\big),\quad \omega \in\Omega\setminus N.
    \]
    Then, preference consistency of $\rrm$ implies that $\rrm(\mathbf{X}) = \rrm(\mathbf{Y})$. This shows that $\rrm_0$ is well-defined. Next, we prove that $\rrm_0$ satisfies the corresponding properties. To show that $\rrm_0$ is increasing, let $X,Y\in E_I$ with $X\leq Y$ a.s.~and $\mathbf{X},\mathbf{Y}\in E^n$ with $X = \aggregation(\mathbf{X})$, $Y = \aggregation(\mathbf{Y})$ almost surely. Then, let 
    \[
        N = \{X>Y\}\cup\{X\neq \aggregation(\mathbf{X})\}\cup\{Y\neq \aggregation(\mathbf{Y})\}. 
    \]
    For all $\omega\in \Omega\setminus N$ we have
    \[
        \rrm\big(\mathbf{X}(\omega)\big) = \aggregation(\mathbf{X})(\omega) = X(\omega)\leq Y(\omega) = \aggregation(\mathbf{Y})(\omega) = \rrm\big(\mathbf{Y}(\omega)\big). 
    \]
    Thus, by preference consistency of $\rrm$, we obtain that
    \[
        \rrm_0(X) = \rrm(\mathbf{X})\leq \rrm(\mathbf{Y}) = \rrm_0(Y). 
    \]
    Now, let $x\in I$ and $\bar{x}\in(0,\infty)^n$ such that $\aggregation(\bar{x})= x$, where the existence of such an $\bar{x}$ follows by $\aggregation$ being surjective. Then, we obtain
    \[
        \rrm_{0}(x) = \rrm(\bar{x}) = \aggregation(\bar{x}) = x,
    \]
    i.e.,~$\rrm_0$ satisfies the constancy property on $I$. 
    
    Now, assume $\rrm$ being risk GG-convex\footnote{Note, to prove the ``only if'' part regarding $\rrm_0$, we do not make use of the GG-convexity of $\rrm$.} and choose $X,Y\in E_I$ and $\lambda \in (0,1)$. Moreover, set~$Z = X^{\lambda}Y^{1-\lambda}$ and choose $\mathbf{X},\mathbf{Y},\mathbf{Z}\in E^n$ such that
    \[
        X = \aggregation(\mathbf{X}),\quad Y = \aggregation(\mathbf{Y}),\quad Z = \aggregation(\mathbf{Z})\quad \text{a.s.}
    \]
    Then, there exists a null set $N\in \mathcal{F}$ such that for all $\omega\in\Omega\setminus N$ it holds that
    \begin{align*}
        \rrm\big(\mathbf{Z}(\omega)\big) = \aggregation(\mathbf{Z})(\omega) &= Z(\omega)\\
        &= \big(X(\omega)\big)^{\lambda} \big(Y(\omega)\big)^{1-\lambda}\\
        &= \big(\aggregation(\mathbf{X})(\omega)\big)^{\lambda}\big(\aggregation(\mathbf{Y})(\omega)\big)^{1-\lambda}\\
        &= \rrm\big(\mathbf{X}(\omega)\big)^{\lambda} \rrm\big(\mathbf{Y}(\omega)\big)^{1-\lambda}.
    \end{align*}
    So, by risk GG-convexity of $\rrm$ it holds that
    \[
        \rrm_0(Z) = \rrm(\mathbf{Z}) \leq \rrm(\mathbf{X})^{\lambda}\rrm(\mathbf{Y})^{1-\lambda} = \rrm_0(X)^{\lambda}\rrm_0(Y)^{1-\lambda},
    \]
    and hence, $\rrm_0$ is GG-convex. Finally, assume that $\rrm$ satisfies positive homogeneity. Let $X\in E_I$ and $\lambda\in(0,\infty)$. Then, by positive homogeneity of $\aggregation$, we obtain $\lambda X = \aggregation(\lambda\mathbf{X})$ for all $\mathbf{X}\in E^n$ with~$\aggregation(\mathbf{X}) = X$. For such a random vector $\mathbf{X}$, we obtain
    \[
        \rrm_0(\lambda X) = \rrm(\lambda\mathbf{X}) = \lambda\rrm(\mathbf{X}) = \lambda\rrm_0(X),
    \]
    where we used the positive homogeneity of $\rrm$ in the second equality. Concluding, $\rrm_0$ satisfies all of the desired properties and by definition of $\rrm_0$, we have $\rrm(\mathbf{X}) = (\rrm_0\circ \aggregation)(\mathbf{X})$ for all $\mathbf{X}\in E^n$.\newline

    To prove the ``if'' part, we now assume a return aggregation function $\aggregation$ and a map $\rrm_0$, satisfying the properties in (ii), such that $\rrm = \rrm_0\circ \aggregation$. First, note that $\rrm$ is increasing as composition of two increasing functions. Moreover, $\rrm$ satisfies
    \[
        \rrm\big((0,\infty)^n\big) = \rrm_0\Big(\aggregation\big((0,\infty)^n\big)\Big) = \rrm_0\big(I\big) = I,
    \]
    where we used surjectivity of $\aggregation$ to obtain the second equality and the constancy property on $I$ of~$\rrm_0$ to obtain the third equality. Hence, $\rrm$ is $I$-surjective on constants. Then, it also holds that 
    \[
        \rrm|_{(0,\infty)^n}(E^n) = \big(\rrm_0\circ\aggregation\big)|_{(0,\infty)^n}(E^n) = \rrm_0|_{I}\big(\aggregation(E^n)\big) =\rrm_0|_{I}(E_I) = E_I, 
    \]
    where we used the constancy property on $I$ of $\rrm_0$ in the last equation. To show the preference consistency, let $\mathbf{X},\mathbf{Y}\in E^n$ and let $N\in \mathcal{F}$ be a null set such that for all $\omega\in \Omega\setminus N$ it holds
    \[
        \rrm\big(\mathbf{X}(\omega)\big)\leq \rrm\big(\mathbf{Y}(\omega)\big).
    \]
    The constancy property on $I$ of $\rrm_0$ then gives us for all $\omega\in \Omega\setminus N$ that
    \begin{align*}
        \aggregation(\mathbf{X})(\omega) = \rrm_0\big(\aggregation(\mathbf{X})(\omega)\big) = \rrm\big(\mathbf{X}(\omega)\big)\leq \rrm\big(\mathbf{Y}(\omega)\big) =\rrm_0\big(\aggregation(\mathbf{Y})(\omega)\big) = \aggregation(\mathbf{Y})(\omega).
    \end{align*}
    This means that $\aggregation(\mathbf{X})\leq \aggregation(\mathbf{Y})$ a.s.~and together with $\rho_0$ being increasing, we obtain
    \[
        \rrm(\mathbf{X}) = \rrm_0\big(\aggregation(\mathbf{X})\big)\leq\rrm_0\big(\aggregation(\mathbf{Y})\big) = \rrm(\mathbf{Y}).  
    \]
    Therefore, $\rrm$ is preference consistent.
    Now, assume $\aggregation$ and $\rrm_0$ to be GG-convex and let $\mathbf{X},\mathbf{Y}\in E^n$ and $\lambda\in(0,1)$. This implies that $\aggregation(\mathbf{X}^{\lambda}\mathbf{Y}^{1-\lambda})\leq \big(\aggregation(\mathbf{X})\big)^{\lambda}\big(\aggregation(\mathbf{Y})\big)^{1-\lambda}$. Hence, monotonicity and GG-convexity of $\rrm_0$ implies that
    \[
        \rrm(\mathbf{X}^\lambda \mathbf{Y}^{1-\lambda})\leq \rrm_0\big(\aggregation(\mathbf{X})^{\lambda}\aggregation(\mathbf{Y})^{1-\lambda}\big)\leq\big((\rrm_0\circ\aggregation)(\mathbf{X})\big)^{\lambda}\big((\rrm_0\circ\aggregation)(\mathbf{Y})\big)^{1-\lambda} =  \rrm(\mathbf{X})^{\lambda}\rrm(\mathbf{Y})^{1-\lambda}.
    \]
    This shows that $\rrm$ is GG-convex. To prove risk GG-convexity of $\rrm$, let $\mathbf{X},\mathbf{Y},\mathbf{Z}\in E^n$, $\lambda\in (0,1)$ and~$N\in\mathcal{F}$ with $\P(N) = 0$ and 
    \[
        \rrm\big(\mathbf{Z}(\omega)\big) = \rrm\big(\mathbf{X}(\omega)\big)^{\lambda} \rrm\big(\mathbf{Y}(\omega)\big)^{1-\lambda},\quad \omega\in\Omega\setminus N.
    \]
    By the constancy property on $I$ of $\rrm_0$, this implies that 
    \[
        \aggregation(\mathbf{Z}) = \aggregation(\mathbf{X})^{\lambda}\aggregation(\mathbf{Y})^{1-\lambda}\quad \text{a.s.}
    \]
    Hence, GG-convexity of $\rrm_0$ gives us that
    \begin{align*}
        \rrm(\mathbf{Z}) = (\rrm_0\circ\aggregation)(\mathbf{Z}) &=  \rrm_0\big(\aggregation(\mathbf{X})^{\lambda}\aggregation(\mathbf{Y})^{1-\lambda}\big)\leq \rrm_0\big(\aggregation(\mathbf{X})\big)^{\lambda}\rrm_0\big(\aggregation(\mathbf{Y})\big)^{1-\lambda}= \rrm(\mathbf{X})^{\lambda}\rrm(\mathbf{Y})^{1-\lambda}.
    \end{align*}
    Thus, $\rrm$ is risk GG-convex. Finally, if $\aggregation$ and $\rrm_0$ are positively homogeneous, then for~$\mathbf{X}\in E^n$ and $\lambda\in(0,\infty)$ we obtain that
    \[
        \rrm(\lambda\mathbf{X}) = \rrm_0\big(\lambda\aggregation(\mathbf{X})\big) = \lambda(\rrm_0\circ\aggregation)(\mathbf{X}) = \lambda\rrm(\mathbf{X}).
    \]
    Hence, $\rrm$ satisfies positive homogeneity.
\end{proof} 

\begin{remark}
    Theorem~\ref{thm:representationSystemicRRM} is the RRM counterpart of similar decompositions for classical systemic risk measures. The starting point is~\cite[Theorem 1]{chen2013SystemicRiskMeasures}, where the domain of the systemic risk measure is finite-dimensional and both convexity and positive homogeneity are imposed. Then,~\cite[Theorem 1 and Corollary 1]{kromer2016} are the extensions of this result to systemic risk measures defined on locally convex-solid Riesz subspaces of the space of all measurable functions on a measurable space $(\Omega,\mathcal{F})$.
\end{remark}

Theorem~\ref{thm:representationSystemicRRM} allows us to effectively construct systemic RRMs, by choosing an RRM on $E_I$ and a return aggregation function. The next result shows that classical aggregation functions can be used to obtain return aggregation functions, thereby providing a wide range on possibilities to design return aggregation functions. 

\begin{corollary}\label{cor:constructionAggregation}
    Let $\Lambda:\R^n\rightarrow\log(I)$ be an increasing, surjective (and convex) function such that~$\Lambda(\log(E^n)) = \log(E_I)$. Then,
    \[
        \aggregation:(0,\infty)^n \rightarrow I,\ x\mapsto (\exp\circ\Lambda\circ\log)(x)
    \]
    is a (GG-convex) return aggregation function.
\end{corollary}

\begin{proof}
    The map $\aggregation$ as composition of three increasing functions is increasing. Moreover, 
    \[
        \aggregation\big((0,\infty)^n\big) = \exp\big(\Lambda(\R^n)\big) = \exp(\log(I)) = I
    \]
    and
    \[
        \aggregation(E^n) = \exp\bigg(\Lambda\big(\log(E^n)\big)\bigg) = \exp(\log(E_I)) = E_I. 
    \]
    Hence, $\aggregation$ is a return aggregation function. If $\Lambda$ is convex, then for all $x,y\in(0,\infty)^n$ and $\lambda\in (0,1)$ it holds that
    \begin{align*}
        \aggregation(x^\lambda y^{1-\lambda}) &= \exp\Big(\Lambda\big(\lambda\log(x)+(1-\lambda)\log(y)\big)\Big)\\
        &\leq \exp\Big(\lambda\Lambda\big(\log(x)\big)\Big)\exp\Big((1-\lambda)\Lambda\big(\log(y)\big)\Big) = \aggregation(x)^{\lambda}\aggregation(y)^{1-\lambda},  
    \end{align*}
    and hence, $\aggregation$ is GG-convex.
\end{proof}

The proof of Theorem~\ref{thm:representationSystemicRRM} heavily relies on the untypical conditions of preference consistency and risk GG-convexity. To obtain a sophisticated picture, we illustrate situations in which GG-convex RRMs do not satisfy either preference consistency or risk GG-convexity. We start with a GG-convex RRM, which does not satisfy preference consistency. Such an example is absent even from the seminal paper of~\cite{chen2013SystemicRiskMeasures}, thereby filling an important gap in the literature.
\begin{example}\label{exam:notPreferenceConsistent}
    Assume $I=(0,\infty)$ and $\rrm:(L_{++}^{\infty})^2\rightarrow(0,\infty)$, defined by $\rrm(\mathbf{X}) = \sqrt{\E[X_1]\E[X_2]}$. Obviously, this map is increasing and positively homogeneous. In addition, for all $x\in(0,\infty)$ it holds that  $\rrm((x,x)^{\intercal}) = x$, i.e.,~$\rrm\big((0,\infty)^2\big) = (0,\infty)$, and for every $X\in L_{++}^{\infty}$ we have $\rrm|_{(0,\infty)^2}((X,X)^{\intercal}) = X$, i.e.,~$\rrm|_{(0,\infty)^2}\big((L_{++}^{\infty})^{2}\big)=L_{++}^{\infty}$. Thus, properties (1) and (2) in Definition~\ref{def:systemicRRM} are satisfied. Next, for $\mathbf{X},\mathbf{Y}\in (L_{++}^{\infty})^{2}$ and $\lambda\in(0,1)$ it holds by Hölder's inequality that $\rrm\big(\mathbf{X}^{\lambda}\mathbf{Y}^{1-\lambda}\big)\leq \rrm(\mathbf{X})^{\lambda}\rrm(\mathbf{Y})^{1-\lambda}$. Thus, $\rrm$ is a GG-convex RRM. However, $\rrm$ does not satisfy preference consistency. Indeed, let~$\mathbf{X}$ satisfy $\mathbb{P}\big(\mathbf{X} = (10,0.1)^{\intercal}\big)=\mathbb{P}\big(\mathbf{X} = (0.1,10)^{\intercal}\big)=0.5$ and let $\mathbf{Y} = (1,1)^{\intercal}$. Then, for almost every~$\omega\in\Omega$ it holds that 
    \[
        \rrm(\mathbf{X}(\omega)) = \sqrt{10\cdot 0.1} = 1 = \sqrt{1\cdot 1} =  \rrm(\mathbf{Y}(\omega)),
    \]
    but $\rrm(\mathbf{X}) = \sqrt{5.05\cdot 5.05} = 5.05 > 1 = \rrm(\mathbf{Y})$, showing that $\rrm$ is not preference consistent. The latter property is used in the proof of Theorem~\ref{thm:representationSystemicRRM} to ensure that $\rho_0$ defined in~\eqref{eq:proofSystemicDefinitionSingleRM} is indeed well-defined. 

    In this example, losses are not aggregated before the information contained in the random variables is reduced to real values. Instead, the order is reversed: the reduction is performed first via the expectation operator, yielding scalar quantities, which are then aggregated by taking their product. Consequently, preference consistency necessitates an accumulation step through a return aggregation function $\aggregation$ prior to reducing the information contained in the random variables via $\rrm_0$.
\end{example}

Our next goal is to show that every GG-convex and positively homogeneous RRM satisfies risk GG-convexity. This expands the observation for systemic risk measures in~\cite{chen2013SystemicRiskMeasures}, namely that convexity and cash-additivity is sufficient to obtain risk convexity. The latter is the analogue of risk GG-convexity for monetary risk measures.
\begin{proposition}\label{prop:riskGGconvex}
    Let $\rrm:E^n\rightarrow (0,\infty)$ be a systemic RRM, satisfying GG-convexity and positive homogeneity. Then, $\rrm$ satisfies risk GG-convexity.
\end{proposition}

\begin{proof}
    By applying Theorem~\ref{thm:representationSystemicRRM}, there exist an increasing map $\rrm_0:E_I\rightarrow (0,\infty)$ satisfying the constancy property on $I$ and a return aggregation function $\aggregation:(0,\infty)^n\rightarrow I$ such that $\rrm = \rrm_0\circ \aggregation$. By recalling the proof of Theorem~\ref{thm:representationSystemicRRM}, we can assume without loss of generality that $\aggregation$ is positively homogeneous. Hence, for all $X\in E_I$, it holds that
    \[
        X = \aggregation\bigg(\frac{X}{c_{\mathbf{1}}}\cdot \mathbf{1}\bigg),\quad \text{a.s.},
    \]
    where $c_{\mathbf{1}} = \aggregation(\mathbf{1})\in I$ and $\mathbf{1} = (1, \dots, 1)^{\intercal}\in\R^n$. Together with $\rrm$ being GG-convex, we obtain for all~$X,Y\in E_I$ and $\lambda\in(0,1)$ that
    \begin{align*}
        \rrm_0\big(X^{\lambda} Y^{1-\lambda}\big) 
        = \rrm\Bigg(\bigg(\frac{X}{c_{\mathbf{1}}}\bigg)^{\lambda}\bigg(\frac{Y}{c_{\mathbf{1}}}\bigg)^{1-\lambda}\cdot \mathbf{1}\Bigg) 
        \leq \rrm\bigg(\frac{X}{c_{\mathbf{1}}}\cdot\mathbf{1}\bigg)^{\lambda}\rrm\bigg(\frac{Y}{c_{\mathbf{1}}}\cdot\mathbf{1}\bigg)^{1-\lambda}
        = \rrm_0(X)^{\lambda}\rrm_0(Y)^{1-\lambda}. 
    \end{align*}
    Hence, $\rrm_0$ is GG-convex, which enables us to prove risk GG-convexity of $\rrm$ by arguing analogously to the proof of Theorem~\ref{thm:representationSystemicRRM}. Indeed, let $\mathbf{X},\mathbf{Y},\mathbf{Z}\in E^n$, $\lambda\in (0,1)$ and $N\in\mathcal{F}$ with $\P(N) = 0$ such that for all $\omega\in\Omega\setminus N$ it holds that $\rrm\big(\mathbf{Z}(\omega)\big) = \rrm\big(\mathbf{X}(\omega)\big)^{\lambda} \rrm\big(\mathbf{Y}(\omega)\big)^{1-\lambda}$.
    By the constancy property on $I$ of $\rrm_0$, we obtain that $\aggregation(\mathbf{Z}) = \aggregation(\mathbf{X})^{\lambda}\aggregation(\mathbf{Y})^{1-\lambda}$ almost surely. Together with GG-convexity of~$\rrm_0$ it is straightforward to show that $\rrm(\mathbf{Z})\leq \rrm(\mathbf{X})^{\lambda}\rrm(\mathbf{Y})^{1-\lambda}$, i.e.,~$\rrm$ satisfies risk GG-convexity.
\end{proof}

\begin{remark}\label{rem:riskGGconvex}
    Without positive homogeneity it is not clear yet, if the remaining properties are enough to imply risk GG-convexity. Even the works of~\cite{chen2013SystemicRiskMeasures} and~\cite{kromer2016} do not give an answer to this question in the case of convex systemic risk measures (without satisfying cash-additivity). A situation in which risk GG-convexity still holds is the one of a GG-affine aggregation function $\aggregation$, i.e.,~for all $x,y\in\R^n$ and all $\lambda\in(0,1)$ it holds that $\aggregation(x^{\lambda}y^{1-\lambda}) = \aggregation(x)^{\lambda}\aggregation(y)^{1-\lambda}$, e.g.,~satisfied if $\aggregation$ is the geometric mean. If $\aggregation$ is GG-affine, then for arbitrary $X,Y\in E_I$, there exist $\mathbf{X},\mathbf{Y}\in E^n$ with $X = \aggregation(\mathbf{X})$ and $Y = \aggregation(\mathbf{Y})$. Hence, for all $\lambda\in(0,1)$ we obtain that
    \begin{align*}
        \rrm_0\big(X^{\lambda}Y^{1-\lambda}\big) = \rrm_0\big( \aggregation(\mathbf{X})^{\lambda} \aggregation(\mathbf{Y})^{1-\lambda}\big) = \rrm\big(\mathbf{X}^{\lambda} \mathbf{Y}^{1-\lambda}\big)\leq \rrm(\mathbf{X})^{\lambda}\rrm(\mathbf{Y})^{1-\lambda} = \rrm_0(X)^{\lambda}\rrm_0(Y)^{1-\lambda},
    \end{align*}
    i.e.,~$\rrm_0$ is GG-convex. Showing that $\rrm$ satisfies risk GG-convexity works  analogously to the proof of Proposition~\ref{prop:riskGGconvex}.
\end{remark}

By our Theorem~\ref{thm:representationSystemicRRM} and Remark~\ref{rem:riskGGconvex}, a possibility in which a systemic RRM $\rrm = \rrm_0\circ\aggregation$ satisfies GG-convexity but not risk GG-convexity is the one in which the underlying $\rrm_0$ is increasing, but either the constancy property on $I$ or the GG-convexity fails. Furthermore, we can assume that $\aggregation$ is GG-convex. An overview of such a situation is given in Table~\ref{tab:exampleNonRiskConvex}. The next example follows this guideline to construct a systemic RRM being GG-convex but not risk GG-convex. 

\begin{table}[htb]
\centering
\begin{tabular}{|p{5cm}|p{5cm}|p{5cm}|}
\hline
\rule{0pt}{2.5ex}
\centering $\rrm = \rrm_0\circ\aggregation:E^n\rightarrow(0,\infty)$ & \centering $\rrm_0:E_I\rightarrow(0,\infty)$ & \centering $\aggregation:(0,\infty)^n\rightarrow I$ \tabularnewline
\hline
\begin{itemize}[leftmargin=*, nosep]
    \item GG-convex
    \item Not risk GG-convex
\end{itemize} &
\begin{itemize}[leftmargin=*, nosep]
    \item Not GG-convex
    \item Increasing
    \item Constancy property on $I$
\end{itemize}
&
\begin{itemize}[leftmargin=*, nosep]
    \item GG-convex
    \item Increasing and surjective
    \item $\aggregation(E^n)=E_I$
\end{itemize}
\\
\hline
\end{tabular}
\vspace{0.5em}
\caption{A situation that would lead to a systemic RRM $\rrm$ being GG-convex but not risk GG-convex. Note that this situation atomatically implies that $\rrm$ is a systemic risk measure, i.e.,~$\rrm$ is increasing, $\rrm((0,\infty)^n)= I$ and $\rrm|_{(0,\infty)^n}(E^n) = E_I$.}
\label{tab:exampleNonRiskConvex}
\end{table}

\begin{example}\label{exam:notRiskGGconvex}
    Let $f:\R\rightarrow\R,\ x\mapsto x-e^{-x}$. This function is strictly increasing, strictly concave, twice continuously differentiable and bijective. Note that $f^{-1}$ is surjective, strictly increasing and strictly convex. Then, let $\riskM_0:L^{\infty}\rightarrow \R,\ X\mapsto f^{-1}\big(\E[f(X)]\big)$. By~\cite[Proposition 2.46]{FoeSch}, $f$ not admitting constant absolute risk aversion (CARA) implies that $\riskM$ cannot be cash-additive. Moreover, $\riskM_0$ is obviously increasing and $\riskM_0$ satisfies the constancy property, because for all $x\in\R$ it holds that $\riskM_0(x) = f^{-1}(f(x)) = x$.  

    Then, let $\Lambda:\R^2\rightarrow \R,\ x\mapsto f^{-1}(x_1+x_2)$.\footnote{Note that our example would also work with $\Lambda = f^{-1}$, i.e.,~by imposing a one-dimensional setting. Then, $\Lambda$ would not aggregate losses, because only one loss is present. Hence, only a non-linear transformation is performed. However, by using $\R^2$, we are closer to a systemic risk setup, because we are really aggregating two positions.} By $f^{-1}$ being surjective, increasing and convex, the same properties hold for $\Lambda$. Moreover, given $X\in L^{\infty}$, for $\mathbf{Y}=(f(X), 0)^{\intercal}$ it holds that $\Lambda(\mathbf{Y}) = f^{-1}(f(X)+0) = X$, i.e.,~$L^{\infty}\subseteq \Lambda\big(L^{\infty}(\Omega;\R^2)\big)$. The inverse set inclusion is obvious and hence, $\Lambda\big(L^{\infty}(\Omega;\R^2)\big) = L^{\infty}$. Next, let $\riskM:L^{\infty}(\Omega;\R^2)\rightarrow \R$ be defined by 
    \[
        \riskM(X) = \riskM_0(\Lambda(\mathbf{X})) = f^{-1}(\E[\mathbf{X}_1+\mathbf{X}_2]),\quad \mathbf{X}\in L^{\infty}(\Omega;\R^2),  
    \]
    which is convex as the composition of two convex functions, namely $f^{-1}$ and $\E[.]$. Finally, we have to show that $\riskM$ is not risk convex. By the constancy property of $\riskM_0$, it is enough to find $\mathbf{X},\mathbf{Y}\in L^{\infty}(\Omega;\R^2)$ and $\lambda\in(0,1)$ such that 
    \[
        \riskM_0\big(\lambda\Lambda(\mathbf{X})+(1-\lambda)\Lambda(\mathbf{Y})\big)>\lambda\riskM(\mathbf{X}) + (1-\lambda)\riskM(\mathbf{Y}). 
    \]
    To do so, assume that the underlying probability space is atomless, which ensures the existence of a random variable $Z\in L^{\infty}$ with $\P(Z=1)=\P(Z=-1) = \frac{1}{2}$. Then, set $\mathbf{X} = (y_0+Z,0)^{\intercal}$ and $\mathbf{Y} = (y_0-Z,0)^{\intercal}$, where $y_0\in \R$ is the unique value such that $f^{-1}(y_0) = 0$. Note that $\riskM(\mathbf{X}) = \riskM(\mathbf{Y}) = f^{-1}(y_0) = 0$. Then, by strict convexity of $f^{-1}$ and by $f$ and $f^{-1}$ being strictly increasing, for every $\lambda\in(0,1)$ we obtain that
    \begin{align*}
        \riskM_0\big(\lambda\Lambda(\mathbf{X})+ (1-\lambda)\Lambda(\mathbf{Y})\big) &= f^{-1}\Big(\E\Big[f\big(\lambda f^{-1}(y_0+Z)+ (1-\lambda)f^{-1}(y_0-Z)\big)\Big]\Big)\\
        &>f^{-1}\big(y_0+(2\lambda - 1)\E[Z]\big) \\
        &= f^{-1}(y_0) = 0 = \lambda\riskM(\mathbf{X}) + (1-\lambda)\riskM(\mathbf{Y}). 
    \end{align*}   

    Now, the map $\rrm_0=\exp\circ \riskM_0\circ\log$ on $\interior_{\lVert.\rVert_{\infty}}(L^{\infty}_+)$ is increasing and satisfies the constancy property on $(0,\infty)$, and by Corollary~\ref{cor:constructionAggregation}, the map $\aggregation = \exp\circ\Lambda\circ\log$ is an aggregation function on $(0,\infty)^2$. Hence, we are allowed to apply Theorem~\ref{thm:representationSystemicRRM}, from which we obtain that $\rrm=\rrm_0\circ\aggregation$ is a systemic risk measure on the interior of $L^{\infty}(\Omega;\R^2)$, for which $\rrm = \exp\circ\riskM\circ \log$ holds. From the previous representation of $\rrm$ and the convexity of $\riskM$, we obtain by Proposition~\ref{prop:relationToClassicalRM} that $\rrm$ is GG-convex. It is then easy to check that the missing risk convexity of $\riskM$ implies that $\rrm$ is not risk GG-convex. 
    
    For completeness, note that Proposition~\ref{prop:riskGGconvex} implies that $\rrm$ is not positively homogeneous and missing convexity of $\riskM_0$ can be justified as follows: for $x\in\R\setminus\{0\}$ and $Z\in L^{\infty}$ with $\E[Z] = 0$ set $\phi(t):=\riskM_0(x+tZ)$. Then, it is straightforward to show that
    \[
        \phi^{\prime\prime}(0) = \frac{f^{\prime\prime}(x)}{f^{\prime}(x)}\E[Z^2] <0, 
    \]
    which proves that $\riskM_0$ cannot be convex and hence, $\rrm_0$ cannot be GG-convex.

    Another pair $\riskM_0$ and $\Lambda$ with $\riskM(X) = (\riskM_0\circ\Lambda)(\mathbf{X}) = f^{-1}(\E[\mathbf{X}_1+\mathbf{X}_2])$ would be $\riskM_0(X) = f^{-1}(\E[X])$ and $\Lambda(x) = x_1+x_2$. In this case, $\riskM_0$ is convex, but it does not satisfy the constancy property. Thus, our examples illustrates that risk GG-convexity of a systemic RRM $\rrm$ is essential to guarantee GG-convexity and the constancy property of the corresponding one-dimensional map $\rrm_0$. 
\end{example}

Finally, we illustrate the independence of the map $\rrm_0$ and the return aggregation function $\aggregation$ via a primal and a dual representation. The proof of the primal representation works analogously to the one of~\cite[Proposition 1]{kromer2016} and the proof of the dual representation is based on our Theorem~\ref{thm:dualRepresentation}. 
\begin{proposition}
    Let $(\Omega,\mathcal{F},\P)$ be atomless, $I=(0,\infty)$, $E = \interior_{\lVert.\rVert_{\infty}}(L^{\infty}_+)$, $\rrm:E^n\rightarrow (0,\infty)$ a GG-coherent systemic RRM and $\rrm_0$ resp.~$\aggregation$ the corresponding maps from Theorem~\ref{thm:representationSystemicRRM}.
    \begin{enumerate}
        \item[(i)] The following primal representation holds:
        \[
            \rrm(\mathbf{X}) = \inf\left\{m\in(0,\infty)\,\middle|\,\exists Y\in E:\frac{Y}{m}\in\mathcal{B}_{\rrm_0}\wedge(Y,\mathbf{X})\in\mathcal{C}_{\aggregation}\right\},\quad \mathbf{X}\in E^n,
        \]
        where $\mathcal{C}_{\aggregation} = \{(Y,\mathbf{X})\in E\times E^n\mid \aggregation(\mathbf{X})\leq Y\}$.
        \item[(ii)] If $\aggregation$ is lower semicontinuous, then the following dual representation holds:
        \[
            \rrm(\mathbf{X}) = \max_{l\in L_{\rrm_0}} \max_{h\in L_{\aggregation}} \frac{(l\circ h)(\mathbf{X})}{\Big(\sup\limits_{Z\in\mathcal{B}_{\rrm_0}}l(Z)\Big)\Big(\sup\limits_{z\in\mathcal{B}_{\aggregation}}h(z)\Big)},\quad \mathbf{X}\in E^n,
        \]
        where
        \begin{align*}
            L_{\rrm_0} = \Bigg\{l:X\mapsto \exp\Big(\E_{\Q}[\log(X)] + E_{\nu}[\log(X)]\Big)\,\Bigg|\,&\Big(\frac{\diff \mathbb{Q}}{\diff \mathbb{P}},\nu\Big)=L^1(\Omega;\R^n)_+\oplus (L^1(\Omega;\R^n)^{\disjointComplement})_+,\\
            &\hspace{2cm} \sup\limits_{\substack{Z\geq 0\\ \lVert Z\rVert_{\infty}=1}}\big(\E_{\Q}[Z] + \E_{\nu}[Z]\big) = 1\Bigg\}
        \end{align*}
        and
        \[
            L_{\aggregation} = \Bigg\{x\mapsto \exp\big(\log(x)^{\intercal}\log(y)\big)\,\Bigg|\, y\in(0,\infty)^n,y\geq \mathbf{1},\prod_{i=1}^{n}y_i = e\Bigg\}.
        \]
    \end{enumerate}
\end{proposition}

\begin{proof}
    To prove (i), fix $\mathbf{X}\in E^n$. By Theorem~\ref{thm:representationSystemicRRM} we obtain that
    \begin{align*}
        \rrm(\mathbf{X}) &= \inf\{m\in(0,\infty)\mid \rrm(\mathbf{X})\leq m\} =  \inf\{m\in(0,\infty)\mid (\rrm_0\circ\aggregation)(\mathbf{X})\leq m\}.
    \end{align*}
    Note that for all $Y\in E_I = E$ we obtain
    \[
        \rrm_0(Y) = \inf\{m\in(0,\infty)\mid \rrm_0(Y)\leq m\} = \inf\left\{m\in(0,\infty)\,\middle|\, \frac{Y}{m}\in\mathcal{B}_{\rrm_0}\right\}.
    \]
    Then, by the definition of $\mathcal{C}_{\aggregation}$ it holds that $\big(\aggregation(\mathbf{X}),\mathbf{X}\big)\in\mathcal{C}_{\aggregation}$. Moreover, for all $m\in(0,\infty)$ and~$Y\in E$ with $\frac{Y}{m}\in \mathcal{B}_{\rrm_0}$ and $(Y,\mathbf{X})\in\mathcal{C}_{\aggregation}$, we obtain by monotonicity of $\rrm_0$, that
    \[
        \rrm_0\big(\aggregation(\mathbf{X})\big)\leq \rrm_0(Y)\leq m,
    \]
    and hence, $\frac{\aggregation(\mathbf{X})}{m}\in\mathcal{B}_{\rrm_0}$. Together, we have
    \[
        \left\{m\in(0,\infty)\,\middle|\,\frac{\aggregation(\mathbf{X})}{m}\in\mathcal{B}_{\rrm_0} \right\} = \left\{m\in(0,\infty)\,\middle|\,\exists Y\in E:\frac{Y}{m}\in\mathcal{B}_{\rrm_0}\wedge (Y,\mathbf{X})\in\mathcal{C}_{\aggregation} \right\}. 
    \]
    This implies that
    \begin{align*}
        \rrm(\mathbf{X}) = \inf\left\{m\in(0,\infty)\,\middle|\,\exists Y\in E:\frac{Y}{m}\in\mathcal{B}_{\rrm_0}\wedge(Y,\mathbf{X})\in\mathcal{C}_{\aggregation}\right\}.
    \end{align*}
    To prove (ii), note first that $\rrm_0$ is a GG-convex RRM. Hence, by Theorems~\ref{thm:dualRepresentation} and~\ref{thm:representationSystemicRRM}, we obtain 
    \begin{align}\label{eq:proofDualSystemic1}
        \rrm(\mathbf{X}) = \max_{l\in L_{\rrm_0}}\frac{(l\circ\aggregation)(\mathbf{X})}{\sup\limits_{Z\in\mathcal{B}_{\rrm_0}}l(Z)},\quad \mathbf{X}\in E^n,
    \end{align}
    where the form of the set $L_{\rrm_0}$ follows from~\eqref{eq:setOptimizationDualRepresentation} and by the explanations in Example~\ref{exam:YosidaHewittDecomposition}. Then, note that $\aggregation$ is monotone, positively homogeneous and GG-convex and hence, an RRM. Thus, by Proposition~\ref{prop:relationToClassicalRM}, $\Lambda:=\log\circ\aggregation\circ\exp$ is a proper, lower semicontinuous and convex risk measure. So, by applying the classical Fenchel-Moreau Theorem and similar arguments as in the proof of Theorem~\ref{thm:dualRepresentation} gives us the following for every $x\in\R^n$:
    \begin{align}\label{eq:proofDualSystemic2}
        \aggregation(x) = \max_{\substack{y\in\mathbb{R}^n,y\geq \mathbf{1}\\ \prod_{i=1}^{n}y_i = e}}\frac{\exp\big(\log(x)^{\intercal}\log(y)\big)}{\sup\limits_{z\in\mathcal{B}_{\aggregation}}\exp(\log(z)^{\intercal}\log(y))} = \max_{h\in L_{\aggregation}}\frac{h(x)}{\sup\limits_{z\in\mathcal{B}_{\aggregation}}h(z)}.
    \end{align}
    Then, by noting that every $l\in L_{\rrm_0}$ is increasing and positively homogeneous, the claim in (ii) follows by combining~\eqref{eq:proofDualSystemic1} with~\eqref{eq:proofDualSystemic2}. 
\end{proof}

\begin{remark}
    Given a loss $Y\in E$, then the set $\mathcal{C}_{\aggregation}$ describes all loss vectors $\mathbf{X}\in E^n$ whose aggregated loss $\aggregation(\mathbf{X})$ is no more severe than the loss $Y$. Moreover, our dual representation differs from that in~\cite[Theorem 2]{kromer2016}. We separate the dual elements associated with $\rrm_0$ and $\aggregation$, by  using the acceptance set $\mathcal{B}_{\aggregation}$, but not the set $\mathcal{C}_{\aggregation}$. The trade-off is that the objective function involves the composition $l\circ h$. However, our proof is more direct than the one in~\cite{kromer2016} by relying on Theorem~\ref{thm:dualRepresentation} and the Fenchel-Moreau Theorem.
\end{remark}
    
\section{Separability of vector-valued return risk measures}\label{sec:vectorValuedRRM}

Finally, we study maps $\rrm:E^n\rightarrow (0,\infty]^n$ with $E\subseteq L^{\infty}_{+}$. Motivated by the study of~\cite{AraratFeinstein2024Separability}, we are focusing on the discussion of sufficient conditions such that $\rrm$ can be written as $\rho(\mathbf{X}) = \big(\rrm_1(\mathbf{X}_1),\dots,\rrm_n(\mathbf{X}_n)\big)^{\intercal}$, where each $\rrm_i$ is also an RRM, i.e.,~measuring the risk of the multivariate position $\mathbf{X}$ simplifies to measuring the risk of each component of $\mathbf{X}$ separately. Hence, in such situations, the multidimensionality of the codomain can de facto be ignored. 

We formalize this situation by the separability property of an RRM, which we define next.
\begin{definition}
    Assume a GG-convex cone $E\subseteq L^{\infty}_{+}$ with $(0,\infty)\subseteq E$ and an RRM ${\rrm:E^n\rightarrow [0,\infty]^n}$. Then, we introduce the following properties:
    \begin{enumerate}
        \item[(1)] $\rrm$ is \textbf{separable}, if there exist RRMs $\tilde{r}_1,\dots,\tilde{r}_n$ on $E$ such that $\rrm = (\tilde{r}_1,\dots,\tilde{r}_n)^{\intercal}$.  
        \item[(2)] $\rrm$ is \textbf{pointwise positively homogeneous}, if for all $\mathbf{X}\in E^n$ and $z\in(0,\infty)^n$ it holds that $\rrm(z\mathbf{X}) = z\rrm(\mathbf{X})$, where multiplications are performed pointwise.
        \item[(3)] $\rrm$ satisfies the \textbf{marginal domination property}, if for every $i\in[n]$ there exists a function $\tilde{f}_i:(0,\infty)\rightarrow (0,\infty)$ with $\rrm_i(z)\leq \tilde{f}_i(z_i)$ for all $z\in(0,\infty)^n$. 
    \end{enumerate}
\end{definition}

Under pointwise positive homogeneity, separability of the RRM follows immediately, even without assuming GG-convexity or any additional continuity assumption. Our result also clarifies that --- in case of a monotone and (pointwise) cash-additive functional --- the additional assumptions of convexity and lower semicontinuity in~\cite[Corollary 3.6]{AraratFeinstein2024Separability} are superfluous.
\begin{corollary}\label{cor:separability}
    A pointwise positively homogeneous RRM $\rrm:\interior_{\lVert.\rVert_{\infty}}\big(L^{\infty}_+(\Omega;\R^n)\big)\rightarrow (0,\infty)^n$ is separable. 
\end{corollary}

\begin{proof}
    Let $\mathbf{X}\in\interior_{\lVert.\rVert_{\infty}}\big(L^{\infty}(\Omega;\R^n)\big)$. Without loss of generality, we prove the statement for the first component of $\rrm(\mathbf{X})$. By monotonicity and pointwise positive homogeneity of $\rrm$, we obtain that
    \[
        \rrm_1(\mathbf{X})\geq \rrm_1\big( (\mathbf{X}_1, \essinf \mathbf{X}_2,\dots, \essinf \mathbf{X}_n )^{\intercal}\big) = \rrm_1\big((\mathbf{X}_1, 1,\dots, 1 )^{\intercal}\big). 
    \]
    Analogously, we have that
    \[
        \rrm_1(\mathbf{X})\leq \rrm_1\big( (\mathbf{X}_1, \esssup \mathbf{X}_2,\dots, \esssup \mathbf{X}_n )^{\intercal}\big) = \rrm_1\big((\mathbf{X}_1, 1,\dots, 1 )^{\intercal}\big). 
    \]
    Hence, $\rrm_1(\mathbf{X}) = \tilde{r}_1$, where $\tilde{r}_1:\interior_{\lVert.\rVert_{\infty}}\big(L^{\infty}_{+}\big)\rightarrow (0,\infty)$ is defined by $\tilde{r}_1(X) = \rrm_1\big((\mathbf{X}_1, 1,\dots, 1 )^{\intercal}\big)$.
\end{proof}

The proof strongly relies on pointwise positive homogeneity, which is stronger than positive homogeneity. On the other hand, even under GG-convexity, positive homogeneity does not imply pointwise positive homogeneity, as the next example shows. 

\begin{example}
    For all $\mathbf{X}\in\interior_{\lVert.\rVert_{\infty}}\big(L^{\infty}_+(\Omega;\R^n)\big)$, we set 
    \[
        \rrm(\mathbf{X}) = \Big(\max_{i\in[n]}\,\esssup X_i\Big)\mathbf{1},
    \]
    where $\mathbf{1} = (1,\dots,1)\in\R^n$. We directly see that this is a positively homogeneous RRM. In addition, it is also GG-convex. However, let $n=2$ and assume the constant (random) vector $\mathbf{X} = (1,2)^{\intercal}$. Then, for $z=\Big(2,\frac{1}{2}\Big)^{\intercal}$ it holds that 
    \[
        \rrm(z\mathbf{X}) = \rrm\big((2,1)^{\intercal}\big) = (2,2) \neq (4,1) = \Big(2,\frac{1}{2}\Big)^{\intercal} \rrm\big((1,2)^{\intercal}\big) = z\rrm(\mathbf{X}).
    \]
    Thus, $\rrm$ is not pointwise positively homogeneous.
\end{example}

\begin{remark}
    The marginal domination property is weaker than pointwise positive homogeneity. If an RRM $\rrm:\interior_{\lVert.\rVert_{\infty}}\big(L^{\infty}_+(\Omega;\R^n)\big)\rightarrow (0,\infty)^n$ only satisfies the marginal domination property, then by additionally assuming GG-convexity and $\tau|_{\interior_{\lVert.\rVert_{\infty}}(V_+)}$-lower semicontinuity, where $\tau$ is the topology induced by the sup-norm $\lVert.\rVert_{\infty}$,  Proposition~\ref{prop:relationToClassicalRM} and Remark~\ref{exam:exponentialTopology} gives us that $\riskM =\log\circ\rrm\circ\exp$ is a convex and weak$^{\star}$-lower semicontinuous monetary risk measure. Hence, by applying~\cite[Theorem 3.3]{AraratFeinstein2024Separability} to $\riskM$ we obtain that $\rrm$ is separable with $\rrm(\mathbf{X})=(\tilde{r}_1(\mathbf{X}_1),\dots,\tilde{r}_n(\mathbf{X}_n))^{\intercal}$, where each $\tilde{r}_i$ is a GG-convex RRM.
\end{remark}

We conclude this section by a separability result for a larger domain than $\interior_{\lVert.\rVert_{\infty}}\big(L^{\infty}_+(\Omega;\R^n)\big)$. This result relies on convex duality theory and hence, we have to assume a GG-convex RRM.
To do so, for a GG-convex RRM $\tilde{r}:\interior_{\lVert.\rVert_{\infty}}\big(L^{\infty}_+\big)\rightarrow(0,\infty)$, we define the set 
\[
    \tilde{L}_{\tilde{r}} = \big\{X\in L_{++}^{\infty}\,\big|\,\exists \lambda>0:\tilde{r}^{E}\big(X^{\lambda}\mathbf{1}_{\{X\geq 1\}} + X^{-\lambda}\mathbf{1}_{\{X< 1\}}\big)\leq e\big\},
\]
where 
\[
    \tilde{r}^{E}(X) := \sup_{\mu\in\mathcal{M}_{1,f}}\frac{\exp(\int\log(X)\diff\mu)}{\sup\limits_{Y\in\mathcal{B}_{\tilde{r}}}\exp(\int\log(Y)\diff\mu)},\quad X\in L_{++}^{\infty},
\]
with $\mathcal{M}_{1,f}$ being the set of all $\mu\in\mathbf{ba}_{\P}$ satisfying $\mu(\Omega) = 1$ and $\mu\geq 0$. Note, $\tilde{r}^{E}$ is an extension of~\eqref{eq:dualRRM_interiorCone}, where we used the representations from Example~\ref{exam:YosidaHewittDecomposition}.

\begin{theorem}\label{thm:separability}
    Let $\rrm:\interior_{\lVert.\rVert_{\infty}}\big(L^{\infty}_+(\Omega;\R^n)\big)\rightarrow (0,\infty)^n$ be a separable and GG-convex RRM with $\rrm(\mathbf{X}) = (\tilde{r}_1(\mathbf{X}_1),\dots,\tilde{r}_n(\mathbf{X}_n))^{\intercal}$ for all $\mathbf{X}\in\interior_{\lVert.\rVert_{\infty}}\big(L^{\infty}_+(\Omega;\R^n)\big)\rightarrow (0,\infty)^n$, where  for all $i\in[n]$, $\tilde{r}_i:\interior_{\lVert.\rVert_{\infty}}\big(L^{\infty}_+\big)\rightarrow(0,\infty)$ is a GG-convex RRM. Then, the following statements hold:
    \begin{enumerate}
        \item[(i)] The set $\tilde{L}_{\rrm} := \prod\limits_{i=1}^n \tilde{L}_{\tilde{r}_i}\subseteq L_{++}^{\infty}(\Omega;\R^n)$ is a GG-convex cone.
        \item[(ii)] The map
        \begin{align}\label{eq:monotoneExtension}
            \tilde{f}:L_{++}^{\infty}(\Omega;\R^n)\rightarrow[0,\infty]^n,\ X\mapsto \lim_{m\rightarrow\infty}\rrm(X\vee m^{-1}).
        \end{align}
        is a separable and GG-convex RRM which extends $\rrm$, i.e.,~$\tilde{f}\big|_{\interior_{\lVert.\rVert_{\infty}}\big(L^{\infty}_+(\Omega;\R^n)\big)}=\rrm$.
        \item[(iii)] The restriction $\tilde{f}\big|_{\tilde{L}_{\rrm}}$ is a separable and GG-convex RRM with $\tilde{f}\big|_{\tilde{L}_{\rrm}}>0$ and for all $i\in[n]$ and $\mathbf{X}\in \tilde{L}_{\rrm}$, the following representation holds:
        \begin{align}\label{eq:representationSeparableFunction}
            \big(\tilde{f}\big|_{\tilde{L}_{\rrm}}\big)_i(\mathbf{X}) = \inf\Bigg\{m>0\,\Bigg|\, \frac{\mathbf{X}_i}{m}\in \tilde{\mathcal{B}}_{\tilde{r}_i}\Bigg\}.
        \end{align}
        where $\tilde{\mathcal{B}}_{\tilde{r}_i}=\left\{X\in\tilde{L}_{\tilde{r}_i} \,\middle|\, \tilde{r}_i^{E}(X)\leq 1\right\}$.
    \end{enumerate}
\end{theorem}

\begin{proof}
    To prove (i), note that each $\tilde{L}_{\tilde{r}_i}$ can be written as 
    \[
        \tilde{L}_{\tilde{r}_i} = \{X\in L_{++}^{\infty}\mid  \log(X)\in L_{\tilde{r}_i}\} = L_{++}^{\infty}\cap\exp(L_{\tilde{r}_i}),
    \]
    where
    \[
        L_{\tilde{r}_i} = \{Y\in L^{0}\mid \exists \lambda>0:(\log\circ \tilde{r}_i\circ \exp)(|Y|/\lambda)\leq 1\}.
    \]
    The latter is the Minkowski domain for $r_i=\log\circ \tilde{r}_i\circ \exp$. Then, by~\cite[Proposition 4.10]{owari2014maximumLesbegueExtension}, we have that $ L_{\tilde{r}_i}$ equipped with its Minkowski norm is a Banach lattice, from which (i) follows.

    We prove the properties in (ii) one by one. Monotonicity of $\tilde{f}$ is a direct consequence of monotonicity of $\rrm$. Next, let $\mathbf{X}\in L^{\infty}_{++}(\Omega;\R^n)$ and $\lambda>0$. Then, we obtain that
    \[
        \tilde{f}(\lambda\mathbf{X}) = \lim_{m\rightarrow \infty}\rrm\big((\lambda\mathbf{X})\vee m^{-1}\big) = \lim_{m\rightarrow \infty}\lambda\rrm\big(\mathbf{X}\vee (\lambda m)^{-1}\big).
    \]
    For each $m\in\N$ it holds that 
    \[
        \rrm\big(\mathbf{X}\vee (\lceil\lambda m\rceil)^{-1}\big)\leq \rrm\big(\mathbf{X}\vee (\lambda m)^{-1}\big)\leq \rrm\big(\mathbf{X}\vee (\lfloor\lambda m\rfloor)^{-1}\big),
    \]
    which shows that $\tilde{f}(\lambda \mathbf{X}) = \lambda \tilde{f}(\mathbf{X})$. For GG-convexity, let $\mathbf{X}, \mathbf{Y}\in L_{++}^{\infty}(\Omega;\R^n)$ and $\lambda\in(0,1)$. Then,
    \begin{align*}
        \tilde{f}\big(\mathbf{X}^{\lambda}\mathbf{Y}^{1-\lambda}\big) &= \lim_{m\rightarrow \infty}\rrm\big((\mathbf{X}^\lambda \mathbf{Y}^{1-\lambda})\vee m^{-1}\big)\\
        &\leq\lim_{m\rightarrow \infty}\rrm\Big(\big(\mathbf{X}\vee m^{-1}\big)^\lambda\big( \mathbf{Y}\vee m^{-1}\big)^{1-\lambda}\Big)\\
        &\leq \lim_{m\rightarrow \infty}\rrm\big(\mathbf{X}\vee m^{-1}\big)^\lambda\,\rrm\big(\mathbf{Y}\vee m^{-1}\big)^{1-\lambda}= \tilde{f}\big(\mathbf{X}\big)^{\lambda}\tilde{f}\big(\mathbf{Y}\big)^{1-\lambda}.
    \end{align*}
    Hence, $\tilde{f}$ is GG-convex. Then, by separability of $\rrm$, we obtain that 
    \[
        \tilde{f}(\mathbf{X}) = \begin{pmatrix}
            \lim\limits_{m\rightarrow \infty} \tilde{r}_1(\mathbf{X}_1\vee m^{-1})\\
            \vdots \\
            \lim\limits_{m\rightarrow \infty} \tilde{r}_n(\mathbf{X}_n\vee m^{-1})
        \end{pmatrix},\quad \mathbf{X}\in L_{++}^{\infty}(\Omega;\R^n).
    \]
    Analogously to before, we see that for each $i\in[n]$, the map 
    \[
        X\mapsto \lim\limits_{m\rightarrow \infty} \tilde{r}_i(X\vee m^{-1}),\quad X\in L_{++}^{\infty},
    \]
    is a (GG-convex) RRM, which gives us that $\tilde{f}$ is separable. 
    
    Finally, given $\mathbf{X}\in \interior_{\lVert.\rVert_{\infty}}\big(L_{+}^{\infty}(\Omega;\R^n)\big)$, for $m\in \N$ large enough, we obtain that for all $k\in\N$ with $k\geq m$ it holds that $\rrm(\mathbf{X}\vee k^{-1}) = \rrm(\mathbf{X})$, which proves that $\tilde{f}\big|_{\interior_{\lVert.\rVert_{\infty}}\big(L^{\infty}_+(\Omega;\R^n)\big)}=\rrm$.

    For (iii), separability and GG-convexity of $\tilde{f}\big|_{\tilde{L}_{\rrm}}$ follows by (i) and (ii). For all $\mathbf{X}\in \tilde{L}_{\rrm}$ it holds
    \[
        \big(\tilde{f}\big|_{\tilde{L}_{\rrm}}\big)_i(\mathbf{X}) = \lim_{m\rightarrow\infty} \tilde{r}_{i}(\mathbf{X}_i\vee m^{-1}) = \lim_{m\rightarrow\infty} \exp\Big(r_{i}\big(\log(\mathbf{X}_i)\vee -\log(m)\big)\Big),
    \]
    where $r_i = \log\circ\,\tilde{r}_i\circ \exp$. By defining $r_i^{E} := \log\circ\, \tilde{r}_i^{E}\circ \exp$,~\cite[Theorem~4.5]{Liebrich2017ModelSpaces} implies that 
    \[
        \big(\tilde{f}\big|_{\tilde{L}_{\rrm}}\big)_i(\mathbf{X}) = \exp\big(\inf\big\{m\in\R\,\big|\, r_i^{E}\big(\log(\mathbf{X}_i)-m\big)\leq 0\big\}\big),\quad \mathbf{X}\in \tilde{L}_{\rrm},
    \]
    where the infimum is larger than $-\infty$, as follows from~\cite[Theorem 4.4]{Liebrich2017ModelSpaces}. Hence, $\tilde{f}\big|_{\tilde{L}_{\rrm}}>0$. To prove Equation~\eqref{eq:representationSeparableFunction}, note that the infimum can be rewritten as follows:
    \begin{align*}
        \inf\big\{m\in\R\,\big|\, r_i^{E}\big(\log(\mathbf{X}_i)-m\big)\leq 0\big\} &= \inf\Bigg\{m\in\R\,\Bigg|\, \tilde{r}_i^{E}\bigg(\frac{\mathbf{X}_i}{e^m}\bigg)\leq 1\Bigg\}\\
        &= \log\Bigg(\inf\Bigg\{m>0\,\Bigg|\,\frac{\mathbf{X}_i}{m}\in\tilde{\mathcal{B}}_{\tilde{r}_i}\Bigg\}\Bigg).
    \end{align*}
\end{proof}

\begin{remark}
    Theorem~\ref{thm:separability} shows how to extend the separability result in Corollary~\ref{cor:separability} to larger domains than $\interior_{\lVert.\rVert_{\infty}}(L_{+}(\Omega;\R^n))$. Monotone extensions of risk measures, to which also our function~$\tilde{f}$ in~\eqref{eq:monotoneExtension} counts, has been first studied in~\cite{delbaen2002coherent}. However, only if we restrict to the set $\tilde{L}_{\rrm}$, we can prevent inconsistent risk evaluations, because $\tilde{f}\big|_{\tilde{L}_{\rrm}}>0$. As our proof reveals, the set $\tilde{L}_{\rrm}$ is closely linked to the Minkowski domain of $\riskM$. The latter appears for instance in~\cite{svindland2009subgradients} and~\cite{kupperSvindland2011dual}. 
      
\end{remark}

Finally, we provide the form of the set $\tilde{L}_{\tilde{r}}$ with $\tilde{r} = \exp\circ r\circ \log$ and $r$ being either the entropic risk measure or the Expected Shortfall. In particular, these examples indicate that $L_{++}^{\infty}\setminus \tilde{L}_{\tilde{r}}$ is typically nonempty. 

\begin{example}
    For parameters $\gamma_i>0$, $i\in[n]$, assume the following separable, GG-convex RRM:
    \[
        \rrm(\mathbf{X}) = \big(\tilde{r}_1(\mathbf{X}_1),\dots,\tilde{r}_n(\mathbf{X}_n)\big)^{\intercal}
        = \Big(\E[\mathbf{X_1}^{\gamma_1}]^{\frac{1}{\gamma_1}},\dots,\E[\mathbf{X_n}^{\gamma_n}]^{\frac{1}{\gamma_n}}\Big)^{\intercal},\quad \mathbf{X}\in\interior_{\lVert.\rVert_{\infty}}\big(L_{+}^{\infty}(\Omega;\R^n)\big).
    \]
    By~\cite[Example 3]{bellini2018RRM}, we obtain for each $i\in[n]$ that 
    \[
        r_i(X):=(\log\circ\,\tilde{r}_i\circ\exp)(X) = \gamma_i^{-1}\log(\E[\exp(\gamma_i X)]),\quad X\in\interior_{\lVert.\rVert_{\infty}}(L_{+}^{\infty}),
    \]
    which is an entropic risk measure with parameter $\gamma_i$. Then,~\cite[Example~5.4]{Liebrich2017ModelSpaces}, allows us to calculate for $i\in[n]$ that
    \begin{align*}
        \tilde{L}_{\tilde{r}_i} &= \big\{X\in L_{++}^{\infty}\,\big|\, \exists \lambda>0:\exp\big(\big|\log\big(X^{\lambda}\big)\big|\big)\in L^{1}\big\}\\ &= \big\{X\in L_{++}^{\infty}\,\big|\, \exists \lambda>0:X^{\lambda}\mathbf{1}_{\{X\geq 1\}} + X^{-\lambda}\mathbf{1}_{\{X< 1\}}\in L^{1}\big\}. 
    \end{align*}
    Now, let $Y$ be a random variable with density function $f_Y(y) = (1+y)^{-2}$, $y\geq 0$. Then, set $X=\exp(-Y)\leq 1$ a.s.,~from which we obtain that
    \[
        \E\Big[X^{\lambda}\mathbf{1}_{\{X\geq 1\}} + X^{-\lambda}\mathbf{1}_{\{X< 1\}}\Big] = \E\Big[X^{-\lambda}\Big] = \E[\exp(\lambda Y)] = \int_{0}^{\infty} \frac{e^{\lambda y}}{(1+y)^{2}}\diff y =\infty,
    \]
    i.e.,~$X^{\lambda}\mathbf{1}_{\{X\geq 1\}} + X^{-\lambda}\mathbf{1}_{\{X< 1\}}\notin L^{1}$. Hence, $X\notin \tilde{L}_{\tilde{r}_i}$ for all $i\in[n]$. Concluding, $L_{++}^{\infty}(\Omega;\R^n)\setminus \tilde{L}_{\rrm}\neq \emptyset$.
\end{example}

\begin{example}
    We use the RRM counterpart of the Expected Shortfall as final example, i.e.,~for given probability levels $p_i$, $i\in[n]$, we define $\rrm$ for every $\mathbf{X}\in\interior_{\lVert.\rVert_{\infty}}\big(L_{+}^{\infty}(\Omega;\R^n)\big)$ by
    \begin{align*}
        (\log\circ \rrm)(\mathbf{X}) &= \big((\log\circ\,\tilde{r}_1)(\mathbf{X}_1),\dots,(\log\circ\,\tilde{r}_n)(\mathbf{X}_n)\big)^{\intercal}\\
        &= \Big(\frac{1}{1-p_1}\int_{p_1}^{1}\log(q_{\mathbf{X}_1}(u))\diff u, \dots , \frac{1}{1-p_n}\int_{p_n}^{1}\log(q_{\mathbf{X}_n}(u))\diff u\Big)^{\intercal},  
    \end{align*}
    where $q_X$ denotes the (upper) quantile function of a random variable $X\in L^{0}$. Then, with the help of~\cite[Example 5.5]{Liebrich2017ModelSpaces}, we obtain for each $i\in[n]$ that
    \begin{align*}
        \tilde{L}_{\tilde{r}_i} &= \Big\{X\in L_{++}^{\infty}\,\Big|\, \forall \lambda>0:\tilde{r}_i\Big(\exp\big(\big|\log\big(X^{\lambda})\big|\big)\Big)<\infty\Big\}
        = \Bigg\{X\in L_{++}^{\infty}\,\Bigg|\, \int_{p_i}^{1} q_{|\log(X)|}(u)\diff u <\infty\Bigg\}.
    \end{align*}
    Motivated by this representation, assume a random variable $Y\sim \text{Par}(1,\alpha)$ with shape parameter $\alpha\in(0,1]$, i.e.,~$Y$ is Pareto distributed with infinite mean ($\alpha\leq 1$) and has support $[1,\infty)$. Then, set $X=\exp(-Y)$, from which we obtain that
    \[
        \int_{p_i}^{1} q_{|\log(X)|}(u)\diff u = \int_{p_i}^{1} q_{Y}(u)\diff u = \E\big[Y\mathbf{1}_{\{Y\geq q_{Y}(p_i)\}}\big] \geq \E[Y] = \infty.
    \]
    So, for all $i\in[n]$ we have $X\notin \tilde{L}_{\tilde{r}_i}$ and hence, $L_{++}^{\infty}(\Omega;\R^n)\setminus \tilde{L}_{\rrm}\neq \emptyset$.
\end{example}

\section{Conclusion and outlook}\label{sec:conclusionOutlook}

We revisit our main results and highlight two directions for future research. Addressing the latter would require substantial modifications of the current framework as well as the development of new technical methods, and is therefore beyond the scope of this manuscript. 

\subsection{Conclusion}

In this manuscript, we extend the definition of RRMs to multivariate frameworks. Thereby, we model the domain as an AM-algebra, a novel approach for RRMs. This benefits from treating the possible domains of $(0,\infty)^n$ and $L^{\infty}_{++}(\Omega;\R^n)$ in a unifying manner. Moreover, instead of focusing on a codomain of $(0,\infty)$, we use~$[0,\infty]^m$. 

We initiate the study with epigraphic characterizations of positive homogeneity and GG-convexity. We then relate RRMs to monetary risk measures, yielding a characterization of GG-convexity via convexity of the corresponding monetary risk measure. Under suitable assumptions, RRMs take values in $(0,\infty)$ and they are only locally Lipschitz-continuous; a counterexample shows that Lipschitz-continuity does not hold in general. Then, we obtain two dual representations with different sets of potential optimizers, only one of which guarantees the existence of an optimizer. 

We then discuss two specific situations. First, for domain $L^{\infty}_{++}(\Omega;\mathbb{R}^n)$ and codomain $(0,\infty)$, we introduce the notion of a systemic RRM. We show that every GG-convex systemic RRM admits a representation as a composition of a GG-convex map and a GG-convex return aggregation function, a result that strongly relies on preference consistency and risk GG-convexity. We clarify the roles of these assumptions via two examples, in which either preference consistency or risk GG-convexity are not satisfied. Notably, such examples do not appear in the current literature. However, under positive homogeneity, GG-convexity implies risk GG-convexity. 

Finally, we introduce vector-valued RRMs and study separability, i.e.,~the case in which each RRM component depends only on the corresponding loss component. Restricting to the interior of $L^{\infty}_{+}(\Omega;\mathbb{R}^n)$ as domain, we show that pointwise positive homogeneity implies separability, while the weaker notion of positive homogeneity is insufficient in general; as shown by a counterexample. Finally, we discuss possibilities to extend this result to domains beyond this interior.

\subsection{Outlook: Beyond AM-algebras}

Even if the logarithm of the domain of an RRM does not yield a desirable vector space structure --- such as an~$L^p$-space with $p\in[1,\infty]$ --- it may still be possible to define the RRM via corresponding formulas for monetary risk measures. For instance, even if $\log(X)\notin L^1$, for $X\in L^1_{++}$, it is still possible to calculate the Expected Shortfall (ES) of $\log(X)$ at level $p\in[0,1)$ as
\[
    \es_{p}\big(\log(X)\big) := \frac{1}{1-p}\int_p^1 q_{\log(X)}(u)\diff{u}.
\]
However, for $X,Y\in L^1_{++}(\Omega;\R)$ and $\lambda\in(0,1)$ we do not obtain $X^{\lambda}Y^{1-\lambda}\in L^1_{++}(\Omega;\R)$ in general. Hence, the problem of extending the framework beyond AM-algebras while retaining a well-defined notion of GG-convexity remains open. 

A starting point to answer this question is the fact that every AM-algebra with (approximate) unit is an f-algebra~\cite[Corollary 2.13]{MunozTracedete2025AMalgebras}. An \textbf{f-algebra} is a Riesz algebra with the additional property that for all $x,y\in V_{+}$ with $|x|\wedge |y| = 0$ it follows that $(xz)\wedge y = (zx)\wedge y = 0$ for all $z\in V_+$. A sequence $(x_n)_n$ in a Riesz space $V$ is called relatively uniformly convergent to $x\in V$, whenever there exist some $u\in(0,\infty)$ and $\varepsilon_n\downarrow 0$ in $\R$ such that $|x_n-x|<\varepsilon_n u$ for all $n$. A sequence $(x_n)_n$ in Riesz space $V$ is called uniformly Cauchy whenever there exists some $u\in(0,\infty)$ such that for each $\varepsilon>0$ we have $|x_n-x_m|<\varepsilon u$ for all $n$ and $m$. A Riesz space $V$ is then called uniformly complete if every uniformly Cauchy sequence is relatively uniformly convergent.

To define GG-convexity, we need to introduce powers of positive elements. The following result from~\cite[Corollary 6]{BeukersHuijsmans1984FunctionalCalculus_fAlgebras} allows us to define $x^{r}$ for a positive element~$x$ and $r\in\mathbb{Q}_{+}$. For this, recall that a Riesz space $V$ is Archimedean, if for all $x\in V_{+}$ it holds that $n^{-1}x\downarrow 0$.\footnote{A Riesz space $V$ is Dedekind complete, if every nonempty subset that is bounded from above has a supremum. Every Dedekind complete Riesz space is automatically Archimedean, see~\cite[Lemma~8.4]{AliprantisBorder}.} By~\cite[Theorem 8.21]{AliprantisBurkinshaw}, every Archimedean f-algebra is commutative.

\begin{corollary}
    Let $V$ be an Archimedean uniformly complete f-algebra with unit element. For every $n\in\mathbb{N}$ and $y\in V_{+}$, there exists a unique $x\in V_{+}$ such that $x^{n} = y$. The element $x$ from the previous result is called the $n$-th root of $y$ and we define $y:=x^{\frac{1}{n}}$.    
\end{corollary}

Consequently, for $r = \frac{q}{p}$ with $p,q\in \Q_{+}\setminus\{0\}$ set $x^{r}=(x^p)^{\frac{1}{q}}$. However, by using f-algebras we still lose track of powers for irrational numbers. Therefore, one needs to find out which additional assumptions an f-algebra needs to satisfy such that $y^{t}$ is defined for all $y\in V_+$ and $t>0$ and perform an analysis for RRMs on corresponding f-algebras. In conclusion, studying risk measures on f-algebras appears to be a promising direction for future research to better understand the compatibility between risk measures and algebraic multiplications.

\subsection{Outlook: Set-valued RRMs}

The study of systemic and vector-valued RRMs in Sections~\ref{sec:multivariateRRM} and~\ref{sec:vectorValuedRRM} motivates the study of corresponding set-valued maps. Instead of capital reserves as outputs, they yield allocations that render the vector-valued loss acceptable, thereby focusing on managerial actions required after measuring the risk.
The theory of set-valued risk measures is quite large and it started with the seminal work of~\cite{Jouini2004VectorValued}. Consequences of separability results on set-valued risk measures are discussed in~\cite{AraratFeinstein2024Separability}. Studying dual representations of set-valued RRMs can complement our RRM dual representations. Dual representations of set-valued risk measures can be found in ~\cite{Hamel2010SetValuedDuality,Hamel2011SetValuedDuality,Molchanov2016SetValuedDuality}. The importance of set-valued maps in the context of systemic risk is demonstrated in~\cite{FeinsteinWeberRudloff2017}. Recent developments focuses on star-shaped set-valued risk measures and the consistency between set-valued risk measures and stochastic orders, see~\cite{Nie2025SetValuedStarShaped} and~\cite{Mastrogiacomo2026SetValuedStochasticOrders}.

\appendix

\section{Auxiliary results for Section~\ref{sec:unifyingFramework}}

We collect auxiliary results used in the proofs in the main text. Since these results are standard and not directly related to RRMs, they have been relegated to the appendix. 

\begin{corollary}\label{cor:interiorCK}
    Let $K$ be a compact Hausdorff space. Then,
    \[
        \interior_{\lVert.\rVert_{\infty}}(C(K)_+) =\{f\in C(K)\mid \forall x\in K:f(x)>0\}.
    \]
\end{corollary}

\begin{proof}
    Let $f\in \interior_{\lVert.\rVert_{\infty}}(C(K)_+)$. There exists $\varepsilon>0$ with $\ball_{\lVert.\rVert_{\infty}}(f,\varepsilon)\subseteq C(K)_+$. By a proof of contradiction, assume there exists $x_0\in K$ with $f(x_0)=0$. Then for sufficiently small $\delta>0$ the map $g=f-\delta\, \mathbf{1}_K$ satisfies $\lVert g-f\rVert_{\infty}<\varepsilon$, i.e.,~$g\in \ball_{\lVert.\rVert_{\infty}}(f,\varepsilon)$ and hence, $g\in C(K)_+$. However, $g(x_0)<0$ implies that $g\notin C(K)_+$, a contradiction. Thus,
    \[
        \interior_{\lVert.\rVert_{\infty}}(C(K)_+) \subseteq\{f\in C(K)\mid \forall x\in K:f(x)>0\}.
    \]
    Now, let $f(x)>0$ for all $x\in K$. By $K$ being compact and $f$ being (lower semi-)continuous, $f$ attains a minimum on $K$, see~\cite[Theorem 2.43]{AliprantisBorder}. Hence, 
    \[
        r:=\min_{x\in K} f(x)>0.
    \]
    Then, for $\varepsilon=\frac{r}{2}$ and for each $g\in C(K)$ with $\lVert g-f\rVert_{\infty}<\varepsilon$, we obtain for all $x\in K$ that 
    \[
        g(x)\geq f(x)-\varepsilon\geq r -\frac{r}{2}=\frac{r}{2} >0,
    \]
    that is,~$g\in C(K)_{+}$. This shows that $\ball_{\lVert.\rVert_{\infty}}(f,\varepsilon)\subseteq C(K)_+$ and hence, $f\in \interior_{\lVert.\rVert_{\infty}}(C(K)_+)$.
\end{proof}

\begin{lemma}\label{lem:pushforwardTopology}
    Let $\Phi:X\rightarrow Y$ be a bijection between two nonempty sets $X$ and $Y$ and $\tau$ is a topology on $X$. Then, for $\Phi(\tau)=\{ \Phi(U)\mid U\in\tau\}$ the following statements hold:
    \begin{enumerate}
        \item[(i)] $\Phi(\tau)$ is a topology on $Y$;
        \item[(ii)] For a net $(x_\alpha)_\alpha\subseteq X$ with $x_\alpha\xrightarrow{\tau}x$ it holds that $\Phi(x_\alpha)\xrightarrow{\Phi(\tau)}\Phi(x)$. 
    \end{enumerate}
\end{lemma} 

\begin{proof}
    For part (i), we verify the defining properties of a topology. Obviously, $\Phi(\emptyset) = \emptyset\in\Phi(\tau)$ and by $\Phi$ being surjective, we have $\Phi(X)=Y\in\Phi(\tau)$. Next, let $(V_i)_{i\in I}\subseteq \Phi(\tau)$. Then, for every $i\in I$ there exists $U_i\in\tau$ such that $V_i = \Phi(U_i)$. We obtain that
    \[
        \bigcup_{i\in I} V_i = \bigcup_{i\in I} \Phi\left(U_i\right)  = \Phi\left(\bigcup_{i\in I} U_i\right) \in\Phi(\tau). 
    \]
    So, $\Phi(\tau)$ is closed under arbitrary unions. Finally, let $V_1=\Phi(U_1)$ and $V_2 = \Phi(U_2)$ with $U_1,U_2\in\tau$. Then, $\Phi$ being bijective implies that \[
        V_1\cap V_2 =  \Phi(U_1)\cap \Phi(U_2) = \Phi(U_1\cap U_2) \in\Phi(\tau). 
    \]
    A simple induction shows that $\Phi(\tau)$ is being closed under finite intersections.

    To prove (ii), let $N$ be an arbitrary open neighborhood of $\Phi(x)$ in $(Y,\Phi(\tau))$. There exists $U\in\tau$ with $N=\Phi(U)$. By $\Phi(x)\in\Phi(U)$ and $\Phi$ being injective, it holds that $x\in U$. By assumption, there exists $\alpha_0$ such that for all $\alpha\geq \alpha_0$ we obtain that $x_\alpha\in U$. Applying $\Phi$, we obtain for all $\alpha\geq\alpha_0$ that $\Phi(x_\alpha)\in N$. This means that $\Phi(x_\alpha)\xrightarrow{\Phi(\tau)}\Phi(x)$.  
\end{proof}

\begin{lemma}\label{lem:exponentialOfTopology}
    Let $V$ be an AM-space with order unit and $\tau$ the induced topology by $\lVert.\rVert_{\infty}$. Then, 
    \[
        \exp(\tau) = \tau|_{\interior_{\lVert.\rVert_{\infty}}(V_+)}. 
    \]
\end{lemma}

\begin{proof}
    By the Krein-Kakutani Theorem~\cite[Theorem 9.32]{AliprantisBorder} it is enough to prove the statement for $V=C(K)$. It is a straightforward task to check that 
    \[
        f:(C(K),\tau)\rightarrow (\interior_{\lVert.\rVert_{\infty}}(C(K)_+),\tau|_{\interior_{\lVert.\rVert_{\infty}}(C(K)_+)}),\ g\mapsto \exp(g)
    \]
    is a homeomorphism. Now, let $U\in\tau$. Then, $\exp(U)$ is $\tau|_{\interior_{\lVert.\rVert_{\infty}}(C(K)_+)}$-open due to $f^{-1}$ being continuous. Hence, 
    \[
        \exp(U) = O\cap \interior_{\lVert.\rVert_{\infty}}(C(K)_+)
    \]
    for some $O\in \tau$, proving that $\exp(\tau)\subseteq \tau|_{\interior_{\lVert.\rVert_{\infty}}(C(K)_+)}$. 
    
    Next, let $W=O\cap \interior_{\lVert.\rVert_{\infty}}(C(K)_+)$ for some $O\in \tau$. Then, $f$ being continuous implies that $f^{-1}(W)\in \tau$. By using that $f$ is bijective, we obtain 
    \[
        W = f(f^{-1}(W)) = \exp(f^{-1}(W)) \in \exp(\tau). 
    \]
    Thus, $\tau|_{\interior_{\lVert.\rVert_{\infty}}(C(K)_+)}\subseteq \exp(\tau)$. 
\end{proof}

\begin{lemma}\label{lem:bandDecomposition}
    Assume a Banach lattice $V$ and a band $E$ in $V$ with $V = E\oplus E^{\disjointComplement}$. Then, 
    \[
        V_+ = E_+\oplus (E^{\disjointComplement})_+.
    \]
\end{lemma}

\begin{proof}
    Let $x\in E_+$ and $y\in (E^{\disjointComplement})_+$ and hence, $x\geq 0$ and $y\geq 0$. Then, $x+y\geq 0$ and so, $x+y\in V_+$, i.e.,~$E_+\oplus (E^{\disjointComplement})_+\subseteq V_+$.

    For the remaining set inclusion, assume that $z\in V_+$. Then, there exists $x\in E$ and $y\in E^{\disjointComplement}$ such that $z = x+y$. We have to show that $y\geq 0$ and $z\geq 0$. From~\cite[Section 3]{AliprantisBurkinshaw} we know that there exist corresponding band projections $P_E:V\rightarrow E$ and $P_{E^{\disjointComplement}}:V\rightarrow E^{\disjointComplement}$, which are automatically positive. Hence, by $z\geq 0$ we obtain that
    \[
        x = P_E(z)\geq 0\quad\text{and}\quad y = P_{E^{\disjointComplement}}(z)\geq 0.
    \]
    This shows that $V_+\subseteq E_+\oplus (E^{\disjointComplement})_+$. 
\end{proof}

\section{Positive real line as domain}\label{sec:supplementaryPositiveRealLine}

We add two findings for GG-convex functions. The first one is based on Proposition~\ref{prop:relationToClassicalRM} (ii), from which we obtain the following characterization as counterpart to analogous representations under so-called logarithmic convexity.

\begin{corollary}
    Let $f:(0,\infty)^n\rightarrow (0,\infty)$ be GG-convex and differentiable. For all $x,y\in (0,\infty)^n$ it holds that
    \[
        \exp\left(\frac{\nabla f(x)}{f(x)}\log\left(\frac{y}{x}\right)^{\intercal}\right)\leq \frac{f(y)}{f(x)},
    \]
    where $\frac{y}{x}$ refers to pointwise division.
\end{corollary}

\begin{proof}
    Let $x,y\in(0,\infty)^n$. By Proposition~\ref{prop:relationToClassicalRM}, $f$ is GG-convex if and only if $g = \log\circ f\circ \exp$ is convex on $\R^n$. This means, for $\tilde{x} = \log(x)$ and $\tilde{y}=\log(y)$ it holds that
    \[
        (\tilde{y}-\tilde{x})^{\intercal} \nabla g(\tilde{x}) \leq g(\tilde{y})-g(\tilde{x}). 
    \]
    By noting that $\nabla g(\tilde{x}) = \frac{\nabla f(x)}{f(x)}$, we obtain 
    \[
       \big(\log(y)-\log(x)\big)^{\intercal} \frac{\nabla f(x)}{f(x)}\leq \log(f(y))
-\log(f(x)),
    \]
    from which the claim follows. 
\end{proof}

GG-convexity --- as well as convexity --- is a weaker property than logarithmic convexity, due to the AM-GM inequality. For the AM-GM inequality we refer to~\cite{Aldaz2009AM_GM_inequality}. In the following, we state examples of functions satisfying only convexity or GG-convexity.
\begin{example}
    We seek functions $f:(0,\infty)\rightarrow(0,\infty)$ that satisfy either GG-convexity or convexity, but not both.
    \begin{enumerate}
        \item \textit{GG-convex, but not convex:} Choose $f(x) = x^{\lambda}$ for $\lambda\in(0,1)$. Then, $f$ is obviously GG-convex, but not convex.
        \item \textit{Convex, but not GG-convex:} Choose 
        \[
            f(x) = \left(x-\frac{1}{2}\right)^2 + \frac{3}{4} = x^2-x+1.
        \]
        Then $f$ is obviously convex. However, for $g=\log\circ f\circ \exp$ it holds for all $y\in \R$ that 
        \[
            g^{\prime\prime}(y) = \frac{\e^{y}(-\e^{2y}+4\e^{y}-1)}{(\e^{2y}-\e^{y}+1)^2}.
        \]
        Set, $z=\e^{y}$. Then, for $z>2+\sqrt{3}$ it holds that $g^{\prime\prime}(\log(z))<0$, giving that $g$ is not convex and hence, by Proposition~\ref{prop:relationToClassicalRM} (ii) $f$ is not GG-convex.
    \end{enumerate}
\end{example}

\end{document}